\documentclass[a4paper,onecolumn,11pt,accepted=2023-05-11]{quantum_style}
\pdfoutput=1
\usepackage{amssymb,amsmath,xcolor,empheq,physics}
\numberwithin{equation}{section}
\usepackage{epsfig}
\usepackage{epstopdf}
\usepackage{latexsym}
\usepackage{graphicx}
\usepackage{booktabs}
\usepackage{bbm}
\usepackage{enumitem}
\usepackage[numbers,compress]{natbib}
\usepackage{bm}
\usepackage[T1]{fontenc}
\usepackage[linesnumbered,ruled]{algorithm2e}
\usepackage{amsthm,complexity}

\usepackage{enumitem}

\usepackage{natbib}
\usepackage{fancybox}

\usepackage[margin=20pt,small]{caption}
\usepackage{subcaption}

 \usepackage{url}
 
\usepackage[toc]{appendix}

\usepackage{color}
\usepackage{datetime}
\usepackage[
      colorlinks=false,
      linkcolor=darkblue,  
      urlcolor=blue,    
      filecolor=blue,     
      citecolor=red,
linktocpage=true,
      pdfstartview=FitV,
      bookmarksopen=true,
	  hidelinks
      ]{hyperref}

\newcommand*{\doi}[1]{\href{http://dx.doi.org/#1}{doi: #1}}

\ifpdf
\fi

\newcommand*{\boxedcolor}{black}
\makeatletter
\renewcommand{\boxed}[1]{\textcolor{\boxedcolor}{%
  \fbox{\normalcolor\m@th$\displaystyle#1$}}}
\makeatother

\usepackage{calligra}
\DeclareMathAlphabet{\mathcalligra}{T1}{calligra}{m}{n}

\definecolor{cardinal}{rgb}{0.6,0,0}
\definecolor{darkgreen}{rgb}{0,0.5,0}
\definecolor{golden}{rgb}{0.92, 0.7, 0}
\definecolor{midnight}{rgb}{0, 0, 0.5}
\definecolor{darkblue}{rgb}{0.2, 0, 0.8}

 \def\Im{{\rm Im}}

\def\cB{{\cal B}}

\def\cI{{\cal I}}

\def\cN{{\cal N}}
\def\cO{{\cal O}}
\def\cP{{\cal P}}
\def\cQ{{\cal Q}}

\def\Tr{{\rm Tr}\,}

\def\sl{\frak{sl}}

\newcommand{\DQC}{\mathsf{DQC1}}
\newcommand{\yes}{\mathsf{YES}}
\newcommand{\no}{\mathsf{NO}}

\newcommand{\be}{\begin{equation}} \newcommand{\ee}{\end{equation}}
\newcommand{\bea}{\begin{equation} \begin{aligned}} \newcommand{\eea}{\end{aligned} \end{equation}}
\newcommand{\bmu}{\begin{multline}} \newcommand{\emu}{\end{multline}}

\newcommand\equ[1] {\begin{equation}#1\end{equation}}

\newcommand\eqs[1] {\begin{align}#1\end{align}}
\newcommand\eqss[1] {\begin{align}\begin{split}#1\end{split}\end{align}} 

\renewcommand\( {\left(}
\renewcommand\) {\right)}

\usepackage[cmtip,all]{xy}
\newcommand{\longsquiggly}{\xymatrix{{}\ar@{~>}[r]&{}}}

\def\cH{\mathcal H}
\newtheorem{thm}{Theorem}

\newtheorem{definition}{Definition}
\newtheorem{proposition}{Proposition}

\newtheorem{corollary}{Corollary}

\newcommand{\sharpP}{\mathsf{\#P}}

\renewcommand{\braket}[1]{\langle #1 \rangle}

\usepackage{tikz}
\usetikzlibrary{quantikz}

\newcounter{mycount}
\newcommand\myprob[3]{%
   \stepcounter{mycount}
 \par\noindent   {\bfseries Definition\  \themycount} (#1)\\
   {\bfseries Input}: #2\\
   {\bfseries Problem}: #3\par
}

\newcommand\mypromprob[4]{%
   \stepcounter{mycount}
 \par\noindent  {\bfseries Definition\  \themycount} (#1)\\
   {\bfseries Input}: #2\\
   {\bfseries Promise}: #3\\
   {\bfseries Problem}: #4\par
}

\begin{document}  

\begin{titlepage}

\vspace*{1.5cm}


\title{Complexity of Supersymmetric Systems and the (co)Homology Problem}

\author[1]{Chris Cade}

\author[2]{P. Marcos Crichigno}

\affil[1]{QuSoft \& CWI, Amsterdam, the Netherlands.}

\affil[2]{Blackett Laboratory, Imperial College Prince Consort Rd., London, SW7 2AZ, U.K.}

\maketitle




\abstract{\noindent We consider the complexity of the local Hamiltonian problem in the context of fermionic Hamiltonians with $\mathcal N=2 $ supersymmetry and show that the problem remains $\QMA$-complete. Our main motivation for studying this is the well-known fact that  the ground state energy of a supersymmetric system is exactly zero if and only if a certain cohomology group is nontrivial. This opens the door to bringing the tools of Hamiltonian complexity to study the computational complexity of a large number of algorithmic problems that arise in homological algebra, including problems in algebraic topology, algebraic geometry, and group theory. We take the first steps in this direction by introducing the  {\sc $k$-local Cohomology} problem and showing that it is  $\QMA_1$-hard and, for a large class of instances, is contained in $\QMA$. We then consider the complexity of estimating normalized Betti numbers and show that this problem is hard for the quantum complexity class $\DQC$, and for a large class of instances is contained in $\BQP$.  In light of these results, we argue that it is natural to frame many of these homological problems in terms of finding ground states of supersymmetric fermionic systems. As an illustration of this perspective we discuss in some detail the model of  Fendley, Schoutens, and de Boer  consisting of hard-core fermions on a graph, whose ground state structure encodes $l$-dimensional holes in the independence complex of the graph. This offers a new perspective on existing quantum algorithms for topological data analysis and suggests new ones.}

\end{titlepage}


\setcounter{tocdepth}{2}
\tableofcontents
\newpage



\section{Introduction}

The development of the theory of quantum computation has greatly enriched our understanding of computational complexity. Contemplating the ways in which quantum mechanics can be harnessed to process information has not only led to the discovery of  new (quantum) complexity classes, but also to novel tools for establishing properties of known complexity classes, and revealed that certain computational problems which (most likely) cannot be efficiently solved by classical computers can be efficiently solved by quantum computers \cite{Shor_1997}.\footnote{See \cite{watrous2008quantum} for a review on quantum complexity theory and \cite{drucker2011quantum} for a survey on the application of quantum techniques to classical problems. } 
Determining the complexity of a given computational problem, however,  can be a difficult task and the techniques for doing so typically depend on the sort of problem being considered, e.g., whether it is a problem in linear algebra, topology, or group theory.
Some problems are manifestly quantum mechanical  (e.g., finding the ground state of a quantum Hamiltonian \cite{KitaevBook,2003quant.ph..2079K,2004quant.ph..6180K}) but others may conceal their quantum mechanical nature. A ``prime'' example is that of factoring which, although at face value has little to do with quantum mechanics,  was famously shown by Shor~\cite{Shor_1997} to be efficiently solvable by a quantum computer, while it is believed that it cannot be solved efficiently by a classical computer.

In this paper, which is a follow-up to \cite{Crichigno:2020vue}, we point out that there is a large class of computational problems which at first may not seem  quantum mechanical in nature, but in fact they are.  These problems fall under the mathematical field of~   ``homological algebra'' \cite{weibel_1994,CartanEilenberg+2016}, in particular the ``(co)homology problem'' described below.  Although this can be a rather abstract field, the advantage is that it provides a unified framework encompassing a large number of disparate mathematical problems, including problems of practical interest (see, e.g.,  \cite{KaczynskiTomasz2004Ch/T}).

Homological algebra has its origins in the study of topology, in particular in the problem of describing ``$l$-dimensional holes'' in simplicial complexes, but exploded in its applications and reach when it was realized that the formalism can be applied to algebraic systems (see \cite{WEIBEL1999797} for a history of the subject). Homological algebra has now become a powerful tool used in several branches of mathematics, including algebraic topology, algebraic geometry, and group theory, as well as in physics, including  condensed matter theory, quantum gauge theories, and string theory.

\begin{figure}
\begin{center}
\includegraphics{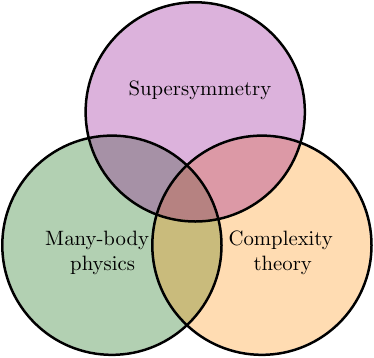}
\end{center}
\caption{The intersection of supersymmetry, many-body physics, and  computational complexity theory considered in this paper. A set of problems in  homological algebra arises at this intersection, whose complexity is naturally captured by various quantum complexity classes. This includes, for instance, the problem of finding $l$-dimensional holes in a simplicial complex.  }
\label{QCT}
\end{figure}

The cohomology problem is described as follows. One is given a vector space, $V$, a ``coboundary'' operator $d:V\to V$ squaring to zero, $d^2=0$, and  is asked to determine the elements $v\in V$ such that $dv=0$, modulo elements of the form $v=du$ for some $u$.
Such elements are, by definition, in the ``cohomology'' of $d$.\footnote{One can similarly define the ``homology'' problem, defined by a ``boundary'' operator $\partial$. These are closely related, with the cohomology problem being the ``dual'' of the homology problem. We will often refer only to the cohomology problem. The precise formulation of these problems in given in Section~\ref{sec:The k-local cohomology problem}. } This includes problems as diverse as finding $l$-dimensional holes in a simplicial complex (via simplicial cohomology), studying fixed points under the action of a group (via group cohomology), and  problems in number theory (via Galois cohomology).

One of our main goals here is to bring the tools of Hamiltonian complexity to  bear on the  complexity of problems in homological algebra, in particular the (co)homology problem just described. The field of  Hamiltonian complexity lies at the intersection of many-body physics and computational complexity. Starting with with Kitaev's celebrated result that determining  the ground state energy of a quantum local Hamiltonian is $\QMA$-complete  \cite{KitaevBook,2003quant.ph..2079K,2004quant.ph..6180K},  it has led to a number of important insights and tools in quantum complexity theory.

The bridge between  Hamiltonian complexity and the cohomology problem is provided by the study of Hamiltonians describing {\it supersymmetric systems}, a class of quantum mechanical systems which are invariant under the exchange of fermionic and bosonic states \cite{Nicolai:1976xp, Witten:1981nf,Witten:1982df,Witten:1982im}. By definition, the Hamiltonian for these quantum mechanical systems  can be written as $H=\{\cQ,\cQ^\dagger\}$, where $\cQ$ is known as the ``supercharge'' of the system, an operator sending bosonic states into fermionic states (and vice versa) and squaring to zero, $\cQ^2=0$.\footnote{Technically, this corresponds to an $\cN=2$ model; see Section~\ref{sec:Supersymmetric Quantum Mechanics} for details. The symbol $\{\cdot,\cdot\}$ denotes the anticommutator: $\{a,b\} := ab + ba$. } An important role is played by states with exactly zero energy, $E=0$, called {\it supersymmetric ground states}. As noted by Witten \cite{Witten:1982df,Witten:1982im} these organize themselves into certain cohomology groups and thus  the question of whether the cohomology group is nontrivial is translated into the physical question of whether the spectrum contains  any states with $E=0$. In this way, the complexity of determining cohomology groups can be studied using the tools of  Hamiltonian complexity, applied to supersymmetric systems. 

A new set of questions thus arises at the intersection of many-body physics, the theory of computational complexity, and supersymmetry (see Fig.~\ref{QCT}) in what may be called ``quantum computational (co)homology.''

To make this concrete, let us illustrate the interplay at this intersection with the beautiful fermion hard-core Hamiltonian  of  Fendley,  Schoutens, and  de Boer \cite{Fendley:2002sg}. The Hamiltonian describes a set of fermionic modes on the vertices $i$ of a simple graph, $G=(V,E)$, created and annihilated by standard fermionic operators $a_i^\dagger$ and $a_i$, respectively.\footnote{The operators $a_i^\dagger$ and $a_i$ create and destroy a fermionic mode (or simply a fermion) at site $i$, respectively. These satisfy the standard anticommutation relations--see Section~\ref{sec:Supersymmetry in lattice systems} for details or, e.g., \cite{bravyi2002fermionic}.}  A fermion is allowed to occupy an empty vertex provided none of its neighbors is occupied, hence the term ``hard-core.'' The Hamiltonian is given by
\equ{\label{HHC}
H=\sum_{ (i,j)\in E} P_i a_i^\dagger a_j P_j+\sum_{i\in V}P_i\,,
}
where $P_i=\prod_{j|(i,j)\in E} (1-a_j^\dagger a_j)$ is a projector which vanishes if any vertex adjacent  to $i$ is occupied. The first term in \eqref{HHC} is a hopping term among adjacent sites (as long as the hard-core condition is satisfied). The second term is required for supersymmetry; the supercharge is given by $\cQ=\sum_i a_i^\dagger P_i$ as can be easily checked. Given a graph $G$ one can define an associated simplicial complex, known as the independence complex $I(G)$ (see Section~\ref{sec:betti_numbers}). It turns out that supersymmetric ground states $\ket{\Omega_{l+1}}$ of the Hamiltonian \eqref{HHC}  with fermion number $F=l+1$ (i.e., with $l+1$ vertices occupied by a fermionic mode), are  in one-to-one correspondence with $l$-dimensional holes in the corresponding independence complex:
\eqs{
\parbox{0.2\textwidth}{\centering $l$-dimensional holes \\ in $I(G)$ } \qquad \Longleftrightarrow \qquad H\ket{\Omega_{l+1}}=0\,.
}
We will return to this model in Section~\ref{sec:Betti numbers in Topological Data Analysis}. (See \cite{huijse2009supersymmetry} for a nice review of this model, its relation to the independence complex, and references therein for related work in the math literature.) 

Finding $l$-dimensional holes in the independence complex has applications in topological data analysis (TDA)~\cite{wasserman2018topological} and a quantum algorithm for this was presented in  \cite{LloydetalTDA}.\footnote{More precisely, the problem considered there relates to the clique complex, but this is equivalent to the independence complex of the complement graph $\bar G$.} An alternative approach to TDA is thus provided by implementing the fermionic mode Hamiltonian \eqref{HHC} -- gates simulating hopping terms in fermionic models have already been implemented on near-term devices~\cite{Foxen_2020}. We comment on this further in Section~\ref{sec:Betti numbers in Topological Data Analysis}, where we also propose an alternative, variational quantum algorithm for TDA.

 A number of supersymmetric models have been studied in the many-body  literature--see \cite{Nicolai:1976xp,1993hep.th...11138S,Fendley_2003,2004JPhA...37.8937Y,Fendley:2006mj,deGier:2015tpa,Fu_2017} for a non-exhaustive list. Supersymmetric systems can arise at special points in parameter space of systems that are well studied in the computational complexity literature, such as an open XXZ spin chain and adding a boundary magnetic field  with a precise value or the $t$-$J$ model at the point $J=-2t$ \cite{2004JPhA...37.8937Y}. Finding ground states of close cousins of these systems remains  $\QMA$-hard (see, e.g.,  \cite{cubitt2016complexity,o2021electronic}).  It is thus conceivable that such simple supersymmetric systems remain $\QMA$-hard and it would be very interesting to establish this.
 
 The goal of this paper, however, is not to focus on particular supersymmetric Hamiltonians but rather to study the computational complexity associated to a large class of  supersymmetric Hamiltonians and thus of a large class of cohomology problems. We view this as a first step towards a program of determining the computational complexity of a large variety of specific homological problems using the techniques of  Hamiltonian complexity, thus providing an overarching framework for studying the complexity of a variety of disparate computational problems. 
 
As we discuss in more detail below, the perspective just described leads to {\it promise} versions of the (co)homology problem. The notion of promise problems -- decision problems in which the input comes with a certain promise~\cite{goldreich2006promise} -- is central to quantum computation. Indeed, in the black box setting, without such promises quantum algorithms can at most provide polynomial speedups over classical algorithms \cite{beals2001quantum}, and in fact no complete problems are known for the class $\BQP$, only for the class ``$\mathsf{PromiseBQP}$'' (see~\cite{janzing2007simple} for a detailed discussion). It is thus essential when considering problems in  homological algebra (or any other field) from a quantum computational perspective to identify appropriate promises. This is precisely what a (supersymmetric) Hamiltonian complexity perspective accomplishes for homological algebra, translating the promises natural for Hamiltonians, such as properties of the ground state or spectral gap, into promises on the homological problems, allowing us to precisely pinpoint their computational complexity. One may hope that this perspective can lead to the identification of new problems with  provable quantum speed-ups (assuming widely believed conjectures in complexity theory). Indeed, as we discuss in detail below, the problem of estimating the rank of general cohomology groups (or Betti numbers) is one such problem.\footnote{More precisely, we show this is the case for the closely related problem of computing ``quasi'' Betti numbers, which we define in Section~\ref{sec:betti_numbers}.} It would be interesting to study whether this remains the case for particular cohomology problems, such as finding the number of $l$-dimensional holes in a simplicial complex. Although we do not establish the hardness of these particular problems here, the quantum Hamiltonian perspective we advocate provides a concrete setting in which to study this and other similar problems.

\

As illustrated above,  a  subclass of supersymmetric Hamiltonians connect to problems in algebraic topology, thus establishing a close link between certain problems in  topology and quantum computation. A well known relation between topology and quantum computation also arises in the context of topological quantum computation, a model of quantum computation based on the manipulation of anyons \cite{Freedman1998PNPAT,2001quant.ph..1025F} (see \cite{preskillnotes,Lahtinen_2017} for reviews). This has also led to a number of complexity results in computational topology, in particular in relation to the computation of the Jones polynomial \cite{Kitaev_2003,Freedman_2002,2000quant.ph..1108F,aharonov2005polynomial}. The connection here is via the study of topological quantum field theory (TQFT), specifically Chern-Simons theory, and its  relation to the Jones polynomial as uncovered by Witten \cite{Witten:1988hf}. The appearance of topology in supersymmetric quantum mechanics is less explicit, being related to supersymmetric ground states or to a particular class of ``supersymmetric'' operators. See \cite{Crichigno:2020vue} for a discussion of these operators in the context of quantum computation.

\ 

The high degree of symmetry of supersymmetric systems often leads to a strong computational control, sometimes even rendering them exactly solvable (see, e.g., \cite{Fendley_2003,Hagendorf_2014}). What makes supersymmetric systems most interesting from our perspective, however, is not that in some cases they may become ``simple'' from the perspective of a  mathematical physicist  but the fact that they remain computationally hard, implying the hardness the corresponding problems in homological algebra.

Although the main focus of this paper is on the problem of computing (co)homology groups, and estimating their ranks, we expect this interplay to extend beyond these problems. Similarly, we note that the field of  Hamiltonian complexity has led  to a number of interesting complexity results (see \cite{2014arXiv1401.3916G} for a review) and we conjecture that many of these  could be ``ported'' over to interesting statements about homological algebra.

\subsection{Summary of results}

We first show that the local Hamiltonian problem remains $\QMA$-complete for supersymmetric Hamiltonians. This is shown in Section~\ref{sec:The SUSY local Hamiltonian problem}, by reduction from the standard (non-supersymmetric) local Hamiltonian problem. 
This result becomes relevant to the study of (co)homology when considering the perfect completeness, or  $\QMA_1$, version of the problem, which we study in Section~\ref{sec:The k-local cohomology problem}. This is the problem of deciding whether the supersymmetric Hamiltonian has any supersymmetric ground states, i.e., states with {\it exactly} zero energy. Since supersymmetric ground states are in one-to-one correspondence with cohomology groups it thus follows that the problem of deciding whether a cohomology group is trivial or nontrivial is $\QMA_{1}$-hard. As we discuss, the problem is generally contained in $\QMA$ and thus its complexity lies somewhere between $\QMA_{1}$ and $\QMA$.

We then move on to the problem of determining the rank of the cohomology group, or Betti number. This is equivalent to counting supersymmetric ground states which is $\#\BQP$-complete, as we discuss in Section~\ref{sec:betti_numbers}.\footnote{$\#\BQP$ (the counting version of $\QMA$) is equal to $\sharpP$ under weakly parsimonious reductions~\cite{Brown_2011}.}
We then consider the problem of approximating Betti numbers up to some additive accuracy. As we discuss in detail, it is less clear how to precisely pinpoint the complexity of this problem, and following~\cite{gyurik} we thus consider a more general problem, which we call ``quasi Betti number estimation'' and which coincides with the problem of estimating Betti numbers when the input satisfies a certain promise.\footnote{Although we define the problem in terms of estimating Betti numbers, this can easily be framed as a decision problem from which the estimate can be recovered via binary search, and so it makes sense to talk about this problem in the context of complexity classes of decision problems such as $\BQP$, $\DQC$, etc.}
Using the tools from the previous section we show that the general quasi Betti number estimation  problem is $\DQC$-hard, and that there is a $\BQP$ algorithm for this problem and thus its hardness lies somewhere between $\DQC$ and $\BQP$. 
This shows that there exists an efficient quantum algorithm for this problem that cannot be dequantized, unless $\DQC \subseteq \BPP$. 
We stress that the quasi Betti number estimation problem is identical to the Betti number estimation problem only when an additional promise is given on the system, in the form of an inverse polynomial gap in the spectrum.
Unfortunately, it is not clear at the moment whether the quasi Betti number problem remains $\DQC$-hard in this case and thus the complexity of the Betti number problem remains open.\footnote{It is in fact possible to show that the true Betti number problem (counting only exactly-0 eigenvalues) is hard for a ``perfect-completeness version'' of $\DQC$ using techniques from~\cite{brandao_thesis} and an appropriate definition of that class. 
However, it is not clear that this result is very meaningful -- a perfect completeness version of $\DQC$ is not a particularly well-defined model, and it is unclear whether the class would be hard to simulate classically, as we suspect is the case for $\DQC$.} 
Studying this in detail may be of great interest, in particular for the independence (or clique) complex mentioned above in light of  possible implementations of a quantum algorithm for this problem in near-term devices.
We comment further about this in Section~\ref{sec:Betti numbers in Topological Data Analysis}.

\paragraph{Relation to previous work.} The current paper expands on the ideas  of \cite{Crichigno:2020vue}, which considered various aspects at the intersection of supersymmetry and quantum computation. The problem of counting (with signs) the number of supersymmetric ground states was shown to be $\#\P$-complete for a  $6$-local supercharge. The problem of determining the ground state energy of supersymmetric systems was posed, suggesting that the problem may remain $\QMA$-hard as we show here. In an independent set of developments, the question of the complexity of estimating normalized Betti numbers of the independence complex of a graph was considered by Gyurik et al. in~\cite{gyurik}. A relaxed version of the problem, ``approximate Betti number estimation'' (ABNE), was defined but its hardness  was left as an open problem. Instead, the authors showed that the more general problem of determining the ``low-lying spectral density'' (LLSD) of a Hamiltonian is  $\DQC$-hard. The relevance of this result to the problem of estimating normalized ``quasi'' Betti numbers, however, was left open. This question is addressed here, in the context of our other main results, where we show that estimating normalized ``quasi'' Betti numbers for general cohomology groups is $\DQC$-hard. The question on whether this remains the case for the clique complex, however, remains open.

\paragraph{Historical comment. }  Supersymmetry was originally proposed in the late 1960s in the context of string theory and quantum field theory as a possible fundamental symmetry of spacetime, with the consequence that for every fundamental bosonic particle there should be an accompanying fundamental fermionic particle. Since then, research in supersymmetry has grown into playing a prominent role in modern theoretical physics:  it is a powerful tool in charting the space of possible quantum field theories and deriving exact, non-perturbative, results; it is a cornerstone of string/M-theory and the study of quantum gravity via the holographic principle; and has led to various developments in pure mathematics including complex geometry, topology, and representation theory. As  mentioned above,  most important to us is the intimate relation between supersymmetry and cohomology. Returning to its origins, the question of whether supersymmetry is  realized as a fundamental symmetry at the level of particle physics has been an active area of theoretical and experimental research for several decades but still remains open. The class of supersymmetric systems that we consider here, however, are not affected by the answer to this question as supersymmetry appears in these systems not as a fundamental symmetry but as a symmetry of certain many-body systems which can, in principle, be engineered in the lab or as an emergent symmetry of certain condensed-matter systems at low energies. 

\paragraph{Experimental realizations.}

The systems we consider can, at least in principle, be realized in the lab as spin or qubit systems with specially tuned interactions to make them supersymmetric. An example is given by the  XXZ model with a specially tuned boundary magnetic field \cite{2004JPhA...37.8937Y}. See \cite{minar2020kink} and references therein for proposed realizations of the fermion hard-core model in 1D using trapped atoms. Other experimentally realizable systems include relativistic supersymmetric quantum systems  emerging at low energies at the boundary of a topological phase \cite{Grover:2013rc}, or from strongly interacting Majorana zero modes \cite{2015PhRvL.115p6401R}.  Thus, the complexity results  described  here and in \cite{Crichigno:2020vue} can be seen as a first step towards characterizing the computational capabilities of these experimentally realizable systems. As pointed out in \cite{Crichigno:2020vue}, the topological nature of certain subsectors of  supersymmetric quantum mechanics suggests a natural robustness of these systems against a certain type of noise. It would be interesting to study this further but leave this and others interesting questions for future work.

\paragraph{Organization of the paper.}

The paper aims at being accessible to both physicists and computer scientists and is as self-contained as possible. To make it accessible to computer scientists we begin in Section~\ref{sec:Supersymmetric Quantum Mechanics} with a review of supersymmetric quantum mechanics as well as a basic introduction to cohomology, focusing on the relation to supersymmetric quantum mechanics and the study of supersymmetric ground states. To make the paper accessible to physicists, we define the basics concepts of quantum complexity theory, in particular the definition of various quantum complexity classes, as they are needed. We refer to \cite{watrous2008quantum} for a survey of quantum complexity theory, to  \cite{2014arXiv1401.3916G} for a survey of  Hamiltonian complexity, and to  \cite{JunkerSUSYBook,Hori:2003ic} for a  broader review of supersymmetry and its myriad applications in statistical physics and mathematics. In Section~\ref{sec:The SUSY local Hamiltonian problem} we overview basic definitions in quantum complexity theory, define the local Hamiltonian problem for supersymmetric systems, and show that it remains $\QMA$-hard. The tools developed there are the basis for the problems of most interest to us, namely that of determining cohomology groups, which we study in Section~\ref{sec:The k-local cohomology problem}, and that of estimating Betti numbers which we study in Section~\ref{sec:betti_numbers}. We note that although we draw much of our intuition from the  SUSY local Hamiltonian problem, the  $k$-local cohomology problem can be formulated without any reference to supersymmetry and thus for readers who are mostly interested in computational (co)homology it is possible to skip directly to Sections~\ref{sec:The k-local cohomology problem} and~\ref{sec:betti_numbers}, which we have made as self-contained as possible, with little reference to previous sections. In Section~\ref{sec:Betti numbers in Topological Data Analysis} we focus our attention on the fermion hard-core model and its relation to the topic of topological data analysis (TDA), suggesting some approaches towards TDA on near-term devices. Finally, we conclude with a discussion of open problems in Section~\ref{Discussion and Outlook}.


\section{Supersymmetric Quantum Mechanics}
\label{sec:Supersymmetric Quantum Mechanics}

We begin by reviewing basic aspects of supersymmetric (SUSY) quantum mechanics  \cite{Nicolai:1976xp, Witten:1981nf,Witten:1982df,Witten:1982im}, including the supersymmetric algebra,  the structure of ground states and  the role of cohomology. We mostly follow the original exposition in \cite{Witten:1982df,Witten:1982im} and also Chapter 10 of  \cite{Hori:2003ic}.\footnote{See also \cite{JunkerSUSYBook} for a textbook focusing on applications of supersymmetry to statistical mechanics.}

\subsection{Definition, algebra, and spectrum}

The Hilbert space of  any quantum mechanical theory can be decomposed as 
\equ{\label{HFB}
\cH= \cH^{B}\oplus \cH^{F}\,,
}
where $\cH^{B}$ and $\cH^{F}$ are the spaces of bosonic and fermionic states, respectively. These are distinguished by an operator, denoted $(-1)^{F}$, which  acts  as $+1$ on bosonic states and as $-1$ on fermionic states.  By definition, a supersymmetric quantum mechanics is a quantum mechanical system in which there is a set of Hermitian operators, $Q_I$, $I=1,\ldots,\cN$, sending states in $\cH^{B}$ into states in $\cH^{F}$ and vice versa. The operators $Q_{I}$ are the generators of supersymmetry, or ``supercharges,'' and one says that the quantum mechanics has $\cN$ real supercharges.  A basic consistency requirement is that 
\equ{\label{QF}
\{(-1)^{F},Q_I\}=0\,,
}
where $\{\cdot,\cdot\}$ is the anticommutator. In addition, one imposes that the supercharges satisfy the algebra
\equ{\label{qiqjH}
\{Q_I,Q_J\} =\delta_{IJ}\, H\,,\qquad \qquad I,J\in\{1,\ldots,\cN\}\,,
}
where $H$ is, by definition,  the Hamiltonian of the system. In words, an $\cN=1$ supersymmetric Hamiltonian is thus a Hamiltonian that can be written as the {\it square} of a parity-odd, Hermitian, operator $Q$. If there are $\cN$ such independent operators, all anticommuting with each other,  then one says that the systems ``has $\cN$ real supercharges.'' 

It follows directly from the algebra above that
\equ{
[Q_{I},H]=0\,,\qquad [(-1)^{F},H]=0\,,
}
and thus the operators $Q_{I}$ and $(-1)^{F}$ are symmetries of such systems. Another crucial property of supersymmetric systems is that the spectrum is positive semi-definite:
\equ{
\bra{\psi}H\ket{\psi}=\bra{\psi}\{Q_{I},Q_{I}\}\ket{\psi}=2\norm{Q_{I}\ket{\psi}}^{2}\geq 0\,,\qquad \text{for each $I$.}
}
In particular, a state $\ket{\Omega}$ has $E=0$ if and only if it preserves all the supersymmetries, i.e., 
\equ{\label{qomega0}
Q_{I}\ket{\Omega}=0\,, \qquad \qquad \qquad I=1,\ldots,\cN\,.
}
Such states, if they exist, are called {\it supersymmetric ground states} and will play a fundamental role in what follows.

\begin{figure}
\begin{center}
\includegraphics[scale=1.6]{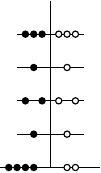}
\put(-85,130){bosonic} \put(-30,130){fermionic}
\put(-130,50){$E>0$}
\put(-130,0){$E=0$}
\caption{The spectrum of supersymmetric quantum mechanics. Black dots represent bosonic states and white ones fermionic states. All states with $E>0$ come in boson/fermion pairs.  Supersymmetric ground states ($E=0$), if there are any, are not necessarily paired.}
\label{fig:spectrum}
\end{center}
\end{figure}
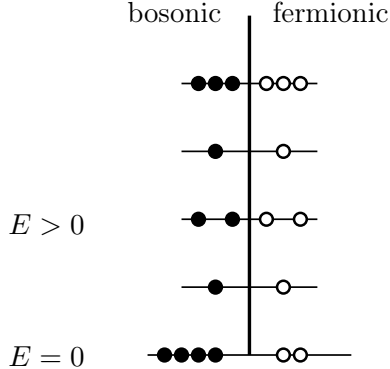

A crucial property of  states with $E>0$ is that they come in boson-fermion pairs with the same energy. To see why this is the case consider for example a bosonic state $\ket{b}$ with $E>0$ and thus there exists at least one supercharge, say  $Q:= Q_{1}$, which does not annihilate the state. Then, the fermionic state  $\ket{f}:= \frac{1}{\sqrt E}Q\ket{b}$ is normalized and has the same energy as $\ket{b}$, since $Q$ commutes with the Hamiltonian. One then says that the states  $(\ket{b},\ket{f})$ form a ``supersymmetric doublet'' and furnish a two-dimensional representation of supersymmetry. When there are several supercharges one may apply multiple $Q_{I}$'s, obtaining a longer multiplet of bosonic and fermionic states but these can be constructed by combining  doublets for each supercharge. The conclusion is that, regardless of the total number of supercharges, for each bosonic state with $E>0$ there is a corresponding fermionic state with the same energy (see Figure~\ref{fig:spectrum}).\footnote{In this paper we shall be mostly interested in finite-dimensional Hilbert spaces and thus the spectrum is discrete.}

The situation for supersymmetric ground states, $E=0$, on the other hand,  is very different. As seen above,  in this case all supercharges annihilate the state, $Q_I\ket{\Omega}=0$, and thus  supersymmetric ground states cannot be paired into  doublets. Rather,  they  stand by themselves, in the trivial representation of supersymmetry. Thus, in general there can be an arbitrary number, $n_{E=0}^{B}$ and $n_{E=0}^{F}$, of bosonic and fermionic supersymmetric ground states, respectively.

Given a supersymmetric system, an interesting question is whether the vacuum state $\ket{\Omega}$ preserves the full supersymmetry of the system, i.e., $\cQ_I\ket{\Omega}=0$ for all $I$. As shown above, this is the case only if the ground state has $E=0$.  When this is not the case, one says that the system breaks supersymmetry ``spontaneously.'' This is in general a rather difficult question to answer (in fact, it is $\QMA_1$-hard as shown below) as this typically requires exact, non-perturbative, methods to determine if the energy of a state is {\it exactly} zero. This led Witten to define the ``index'' \cite{Witten:1982df},
\equ{\label{wittenN=1}
\cI := n^{B}_{E=0}-n^{F}_{E=0}\,,
}
given by the difference in the number of bosonic and fermionic supersymmetric ground states. The quantum mechanical system  of  Figure~\ref{fig:spectrum}, for instance, has Witten index $\cI=4-2=2$.  Under certain conditions, the ``Witten index'' \eqref{wittenN=1} can be computed exactly by perturbative methods (see \cite{Witten:1982df} for a detailed discussion). This simplicity, however, comes at the expense of only providing partial information on supersymmetry breaking; if it is nonzero the  cohomology is necessarily nontrivial but if it vanishes it is inconclusive (it only implies that $n^{B}_{E=0}=n^{F}_{E=0}$ but these could vanish or not). Despite its relative simplicity, computing the Witten index is $\#\P$-hard \cite{Crichigno:2020vue}.

\

It is useful to distinguish operators acting on the Hilbert space according to the action of $(-1)^{F}$ on the operator. An operator  is said to be bosonic or fermionic if it commutes or anticommutes with the parity operator, i.e.,  
\equ{
[(-1)^{F},\cO_{B}]=0\,,\qquad \qquad \{(-1)^{F},\cO_{F}\}=0\,,
}
respectively. It follows that when acting on a state with definite parity, bosonic operators preserve its parity and fermionic operators flip it. The supercharges $Q_{I}$, for example, are fermionic operators and the Hamiltonian is a bosonic operator.

\paragraph{$\cN=2$ systems. } From now on we will focus on systems with $\cN=2$ supersymmetry. In this case it is useful to combine the Hermitian supercharges into the complex combinations 
\equ{
\cQ:= \frac{1}{\sqrt 2} (Q_1+ i Q_2)\,,\qquad \cQ^\dagger:=\frac{1}{\sqrt 2} (Q_1- i Q_2)\,.
}
The supersymmetry algebra \eqref{qiqjH} implies these complex supercharges are nilpotent of degree two, i.e., 
\equ{\label{nilpQ}
\cQ^{2}=(\cQ^{\dagger})^{2}=0\,.
}
The Hamiltonian can be written as
\equ{\label{QQcom}
H=\{\cQ,\cQ^\dagger\}\,.
}
As discussed above, supersymmetric ground states (states with exactly zero energy) must preserve all supersymmetries and thus
\equ{\label{SUSYgs}
\cQ\ket{\Omega}=\cQ^\dagger\ket{\Omega}=0\,.
} 
One of the most important properties of $\cN=2$ supersymmetric quantum mechanics is that supersymmetric ground states are in one-to-one correspondence with the cohomology group of $\cQ$. We review this next.

\subsection{Supersymmetric ground states and cohomology}
\label{sec:Cohomology, homology, and ground states}

A special role is played by states that are ``closed'' or ``exact'' with respect to the supercharge. A state is said to be ``$\cQ$-closed'' if it is in the kernel of $\cQ$, i.e.,
\equ{
    \cQ\ket{\psi} = 0\,.
}
A state is said to be ``$\cQ$-exact'' if it is in the image of $\cQ$, i.e., if there exists some state $\ket{\psi'}$ such that 
\equ{
    \ket{\psi} = \cQ\ket{\psi'}\,.
}
Consider the set of states in the kernel and image of $\cQ$, i.e., 
\equ{
\text{Ker}(\cQ)= \{\ket{\psi}\;|\; \cQ\ket{\psi}=0\}\,,\qquad \text{Im}(\cQ)= \{\ket{\psi}\;|\; \ket{\psi}=\cQ\ket{\psi'}\}\,.
}
The same definitions hold  for  $\cQ^{\dagger}$ and supersymmetric ground states \eqref{SUSYgs} are closed with respect to both $\cQ$ and $\cQ^{\dagger}$. Now, since $\cQ^{2}=0$, all states that are $\cQ$-exact are also $\cQ$-closed and thus $\text{Im}(\cQ)\subseteq\text{Ker}(\cQ)$. Then, the  vector space 
\equ{
H(\cQ):= \frac{\text{Im}(\cQ)}{\text{Ker}(\cQ)}
} 
is well defined and known as the {\it cohomology group} of $\cQ$.\footnote{In this case the cohomology group has the additional structure of a vector space but this is not necessary. See Section~\ref{sec:The k-local cohomology problem} for the general definition. }  It consists of all states in the Hilbert space which are $\cQ$-closed but not $\cQ$-exact. In fact, due to the $\Bbb Z_{2}$-grading of the Hilbert space, $\cH=\cH_{B}\oplus \cH_{F}$, there is a more refined structure and one can define a $\Bbb Z_{2}$-graded ``complex''  of vector spaces,
\equ{\label{complexBF}
C: \qquad \cH^{F}\xrightarrow{\cQ}\cH^{B}\xrightarrow{\cQ}\cH^{F}\xrightarrow{\cQ}\cH^{B}\,.
}
Then, one can define two cohomology groups:
\eqs{
H^{B}(\cQ):=\,& \frac{\text{Ker}(\cQ): \cH^{B}\to \cH^{F}}{\text{Im}(\cQ): \cH^{F}\to \cH^{B}}\,,\\ 
H^{F}(\cQ):=\,& \frac{\text{Ker}(\cQ): \cH^{F}\to \cH^{B}}{\text{Im}(\cQ): \cH^{B}\to \cH^{F}}\,,
}
consisting of all bosonic and fermionic states which are $\cQ$-closed but not $\cQ$-exact, respectively. The full cohomology group is 
\equ{
H(\cQ)=H^{B}(\cQ)\oplus H^{F}(\cQ)\,.
} 
A  complex for which $\text{Ker}(\cQ)=\Im(\cQ)$, and thus $H(\cQ)=0$,  is called an {\it exact} sequence. Thus, cohomology groups are a measure of the failure of a  sequence to be exact.\footnote{For a textbook on algebraic topology see, e.g., \url{http://pi.math.cornell.edu/~hatcher/AT/ATplain.pdf} and for computational aspects see \cite{CompTopBook,zomorodian_2005}.}

Physically, an exact sequence corresponds to a system with no supersymmetric ground states. To see this, note that states with $E>0$ that are $\cQ$-closed  are necessarily $\cQ$-exact, as follows directly from the supersymmetry algebra \eqref{QQcom}. Indeed, applying the algebra  on a state with energy $E$ one has $\{\cQ,\cQ^\dagger\}\ket{\psi}=E\ket{\psi}$. Using that the state is $\cQ$-exact and assuming $E>0$, implies $\ket{\psi}= \cQ \ket{\psi'}$ with $\ket{\psi'}:= \frac{1}{E} \cQ^\dagger \ket{\psi}$ and thus the state $\ket{\psi}$ is  $\cQ$-exact.  On the other hand, all  states with $E=0$ are necessarily closed, as we have already seen, but   {\it cannot} be exact. To show this  suppose for the sake of contradiction that is the case, i.e., that there exists a state $\ket{\psi'}$ such that $\ket{\Omega}=\cQ\ket{\psi'}$. Since $\cQ$ commutes with $H$, the state $\ket{\psi'}$ must also have $E=0$ and hence must be $\cQ$-closed,  $\cQ\ket{\psi'}=0$. The upshot of this analysis is that $\text{Ker}(\cQ)|_{E>0}=\text{Im}(\cQ)|_{E>0}$ while  $\text{Im}(\cQ)|_{E=0}=0$. These two statements together imply
\equ{
 H^B(\cQ)= \cH^B_{E=0}\,,\qquad \qquad  H^F(\cQ)= \cH^F_{E=0}\,,
}
as claimed.  That is, the cohomology of  $\cQ$ is given by the set of supersymmetric ground states. In particular, the total number of SUSY ground states is given by the dimension of the full cohomology group,
\equ{\label{totnsattes}
n_{E=0}^{B}+n_{E=0}^{F}=\dim H^{B}(\cQ) +\dim H^F(\cQ)\,.
}
The Witten index of such systems is then given by 
\equ{\label{Wcoh}
\cI= n_{E=0}^{B}-n_{E=0}^{F}= \dim H^B(\cQ)- \dim H^F(\cQ)\,,
}
which coincides with the ``Euler charactersitic'' of the complex \eqref{complexBF}.

To summarize, finding the space of supersymmetric ground states is equivalent to finding the cohomology group $H(\cQ)=H^{B}(\cQ)\otimes H^{F}(\cQ)$ and computing the Witten index of the system is equivalent to computing the Euler characteristic of the cochain complex. This  will play a crucial role in studying the complexity of the cohomology problem in Section~\ref{sec:The k-local cohomology problem}.

\paragraph{Systems with $F\in \Bbb N_{0}$.} In certain models, a more refined structure may appear. A case of particular interest to us is when there exists an operator $F$   (usually called the fermion number operator)  with integer-valued eigenvalues, $\Bbb N_{0}=\{0,1,\cdots\}$, and thus one can identify the parity operator $(-1)^{F}$, literally as ``minus one to the power $F$, '' or
\equ{
(-1)^{F}=e^{i\pi F}\,.
}
The Hilbert space can then be graded more finely, with respect to eigenvalues  of $F$. We will focus on systems in which the fermion number is bounded above and thus
\equ{
\cH=\cH^{0}\oplus \cH^{1}\oplus\cdots \oplus \cH^{m}\,,
}
where $\cH^{p}$ is spanned by all states with $F=p$ and $m$ is the maximum fermion number of the system. Then $\cH_{B}$ and $ \cH_{F}$ correspond to states with even and odd fermion number, respectively.

By definition, an operator $\cO$ has definite fermion number $f$, if
\equ{
[F,\cO]= f \cO\,.
}
We will informally say that such an operator has fermion number $F=f$. Note that the action by the operator $\cO$ with fermion number $f$ on a state $\ket{s}$ with fermion number $k$ gives a state $\cO\ket{s}$ with fermion number $f+k$; $F \cO\ket{s}=(f+k)\cO\ket{s}$. An operator with $F=0$ thus preserves the fermion number of the states it acts on and the operator is block diagonal in a the basis of states organized by fermion number.

Note that all that is required for defining supersymmetry is that $\cQ$ maps states with even fermion number to states with odd fermion number and vice versa. Thus, the supercharge $\cQ$ must have definite parity, $(-1)^{F}=-1$, but not necessarily definite fermion number $F$ (the supercharge may be given, for instance, by a sum of terms each with different (odd) fermion numbers). In the case that $\cQ$ has  a definite fermion number and equal to $1$, i.e., 
\equ{
[F,\cQ]=\cQ\,,
}
the supercharge acts by sending states in $\cH^{p} $ to states in $\cH^{p+1}$ and the complex \eqref{complexBF} splits into
\equ{\label{complexN}
C:\qquad 0\xrightarrow{\cQ}\cH^{0}\xrightarrow{\cQ}\cH^{1}\xrightarrow{\cQ} \cdots \xrightarrow{\cQ}\cH^{m}\xrightarrow{\cQ}0 \,,
}
where we have included the empty spaces at the ends to emphasize there are no states with $F<0$ or $F>m$. Then one may define the cohomology groups at each level, 
\eqs{
H^{p}(\cQ):=\,& \frac{\text{Ker}(\cQ): \cH^{p}\to \cH^{p+1}}{\text{Im}(\cQ): \cH^{p-1}\to \cH^{p}}\,.
}
Then, each cohomology group $H^p$ cohomology is in one-to-one correspondence with the space of supersymmetric ground states with fermion number $p$:
\equ{
 H^p(\cQ)= \cH^p_{E=0}\,.
}
The $p$th Betti number $\beta_{p}$ is defined as the dimension of the $p$th cohomology group,
\equ{
\beta_{p}:= \dim H^{p}(\cQ)\,.
}
The Witten index of the system is given by the alternating sum
\equ{
\cI=\sum_p (-1)^p \beta_p\,,
}
which coincides with the  Euler characteristic of the complex \eqref{complexN}.

\subsection{Supersymmetry in lattice systems}
\label{sec:Supersymmetry in lattice systems}

A number of interesting models of supersymmetric quantum mechanics have been studied in the literature, with applications ranging from statistical mechanics to string theory and with deep connections to algebraic topology.\footnote{For instance, a supersymmetric version of a free particle moving in a curved manifold leads to a description of de Rham cohomology and turning on a potential leads to Morse theory \cite{Witten:1982im}.}  In this paper we are mostly interested in the realization of $\cN=2$ quantum mechanics in a {\it discrete} system of spins or qubits, as first discussed in \cite{Nicolai:1976xp}. It is most convenient to use the fermionic representation of states. 
In this model one considers $m$ vertices of a graph $G$, each of which can be occupied by 0 or 1 spinless fermions. A fermion at vertex $i$ is created by an operator $a_{i}^{\dagger}$ and annihilated by $a_{i}$,  satisfying the standard  anticommutation relations
\equ{
\{a_i,a_j^\dagger\}= \delta_{ij}\,,\qquad \qquad i,j\in\{1,\ldots,m\}\,.
}
$\ket{0\dots0}$ is the state with no fermionic excitations, which is annihilated by all $a_i$, i.e.,  $a_{i}\ket{0\dots0}=0$ for all $i$. The fermion number operator in Fock space is defined as
\equ{
F=\sum_{i=1}^{m} a^{\dagger}_{i}a_{i}\,.
}
The $2^{m}$-dimensional Fock space is then constructed by acting  with creation operators on the vacuum state, as
\equ{\label{mapstates}
\ket{n_{1}\dots n_{m}}= (a_{1}^{\dagger})^{n_{1}}\cdots (a_{m}^{\dagger})^{n_{m}}\ket{0\dots0}\,.
}
Note that the fundamental modes $a_i^\dagger$ are fermionic. However, the multiparticle states $\ket{n_1\dots n_m}$ can be either bosonic or fermionic  depending on whether the total number of fermionic particles in the state is even or odd, respectively, i.e.,
\eqs{\label{paritystrings}
(-1)^{F}\ket{n_{1}\dots n_{m}}=\begin{cases} +\ket{n_{1}\dots n_{m}}&\text{for $\sum_{i} n_{i}\in 2\Bbb Z$}\\ -\ket{n_{1}\dots n_{m}} &\text{for $\sum_{i} n_{i}\in 2\Bbb Z+1$} \end{cases}
}

Any operator in Fock space can be constructed as a polynomial in the fermionic operators $a_{i}$ and $a_{i}^{\dagger}$. The supercharge $\cQ$ must be a polynomial with only odd degree monomials in order to be consistent with its fermionic nature. In the model of \cite{Nicolai:1976xp} the supercharges are cubic functions of the creation and annihilation operators. More generally one can write
\eqss{
\cQ=\sum_{i}(\alpha_{i}a_{i}^{\dagger}+\bar \alpha_{i}a_{i})+\,&\sum_{ijk}( \alpha_{ijk}a^{\dagger}_{i}a^{\dagger}_{j}a^{\dagger}_{k}+\beta_{ijk}a^{\dagger}_{i}a^{\dagger}_{j}a_{k}+\gamma_{ijk}a^{\dagger}_{i}a_{j}a_{k})  \\
+\,&\sum_{ijk}( \bar \alpha_{ijk}a_{i}a_{j}a_{k}+\bar\beta_{ijk}a_{i}a_{j}a^{\dagger}_{k}+\bar\gamma_{ijk}a_{i}a^{\dagger}_{j}a^{\dagger}_{k})+\cdots
}
where  the $\alpha_{i},\bar \alpha_{i},\alpha_{ijk},\bar \alpha_{ijk}$, etc. are complex coefficients which are constrained  by the nilpotency condition \eqref{nilpQ} and the ellipses denote possible higher power terms. Any such supercharge defines a supersymmetric Hamiltonian via \eqref{QQcom} which, by construction, is parity preserving. If each of the terms above involve at most $k$ fermionic modes we say the supercharge is $k$-local and the resulting Hamiltonian is (at most) $2k$-local. 

Note that the supercharge above has definite parity $(-1)^{F}$ but does not necessarily have a definite fermion number $F$, as it is given by the sum of terms with different odd values of $F$. We will be particularly interested in a large family of supersymmetric systems which arises from imposing two separate conditions on the supercharge.  The first is that the supercharge has definite fermion number $F=1$. We require this for later applications to the cohomology and Betti number problems. Any supercharge with $F=1$ can be written as 
\equ{
\cQ=\sum_{i}a_{i}^{\dagger}\, B_{i}(a,a^{\dagger})\,, 
}
where the $B_{i}$ are bosonic functions with $F=0$ and further constrained by the condition that the supercharge is nilpotent, $\cQ^{2}=0$. The second condition we will impose is that the supercharge is $k$-local and thus each $B^{i}$ above can be written as
\equ{
B_{i}(a,a^{\dagger})=\sum_{S_{i}} \cB_{S_{i}}(a,a^{\dagger})
}
where for each $i$,  $S_{i}\subseteq\{1,\dots,m\}$ is a subset of $m$ vertices of size at most $k-1$, and the $\cB_{S_i}(a,a^{\dagger})$ are bosonic operators all with $F=0$ acting on the sites $S_{i}$. We thus require that each $\cB_{S_{i}}$ is a $(k-1)$-local operator in fermionic space with fermion number $F=0$. 

\paragraph{Examples.} 

The simplest model is  the 1-local supercharge
\equ{
\cQ=\sum_{i}a_{i}^{\dagger}\,,
}
which one can easily check squares to zero and has $F=1$. The resulting Hamiltonian, however, is proportional to the identity operator and is thus rather trivial (in particular, it has no supersymmetric ground states). Examples with a nontrivial Hamiltonian are the original models considered in \cite{Nicolai:1976xp} or the $\cN=2$ version of the SYK model \cite{Fu:2016vas}, given by
\equ{
Q=\sum_{ijk} \alpha_{ijk}a^{\dagger}_{i}a^{\dagger}_{j}a^{\dagger}_{k}\,,
}
where the $\alpha_{ijk}$ is completely antisymmetric and the supercharge has $F=3$. This model has an exponential number of supersymmetric ground states.  

Finally, a very interesting class of models is the fermion hard-core models of \cite{Fendley:2002sg}. These are specified by a choice of a graph $G=(V,E)$ and  
\equ{
\cQ=\sum_{i\in V}a_{i}^{\dagger}\, P_{i}\,,\qquad \qquad P_{i}= \prod_{j\,|\,(i,j)\in E  }(1-\hat n_{j}) \,.
}
The supercharge has $F=1$ and is $(\delta+1)$-local, with $\delta$ the maximum degree of the graph. The resulting Hamiltonian is the one in \eqref{HHC}. We will return to this model in Section~\ref{sec:Betti numbers in Topological Data Analysis} in relation to topological data analysis.

\

In the next few sections we will not focus on particular models but rather on studying the computational complexity of generic supersymmetric systems and its implications for the (co)homology problem.

\section{The SUSY Local Hamiltonian Problem}
\label{sec:The SUSY local Hamiltonian problem}

In this section, we determine the complexity of determining the ground state of supersymmetric systems, i.e., a supersymmetric version of  {\sc $k$-local Hamiltonian} studied by Kitaev.  It was suggested in \cite{Crichigno:2020vue} that the problem remains  $\QMA$-complete for quantum mechanical Hamiltonians with $\cN=2$ supersymmetry and we show here that this is indeed the case.  As we discuss in detail in Section~\ref{sec:The k-local cohomology problem}, a  consequence of this result is that a certain promise version of the cohomology problem is $\QMA_{1}$-hard and in certain cases $\QMA_{1}$-complete (we define $\QMA$ and related complexity classes in Appendix~\ref{app:complexity_classes}).

\subsection{Definitions}
We begin by reviewing {\sc $k$-local Hamiltonian}, defined as follows \cite{KitaevBook}: 
\mypromprob{{\sc  $k$-local Hamiltonian}}{An $n$-qubit local Hamiltonian, $H=\sum_{a} H_{a}$,  where each $H_a$ acts on at most $k$ qubits, and two numbers, $b>a\geq 0$, with  $b-a>\frac{1}{\text{poly}(n)}$.}{The lowest energy level $E_{0}$ of the Hamiltonian is either $E_{0}\leq a$ or $E_{0}\geq b$.}{Output $\yes$ if $E_{0}\leq a$ and $\no$ if $E_{0}\geq b$. \\}

\noindent A well known series of results established that 
{\sc $k$-local Hamiltonian} is $\QMA$-complete  for $k\geq 2$ \cite{KitaevBook,2003quant.ph..2079K,2004quant.ph..6180K}. This is also the case for $2$-local Hamiltonians in fermionic Fock space \cite{Liu_2007}. As reviewed in Section~\ref{sec:Supersymmetric Quantum Mechanics}, supersymmetric Hamiltonians are given in terms of the supercharge $\cQ$ which, due to its anticommuting nature, is most naturally described in fermionic Fock space. We thus formulate the supersymmetric Hamiltonian problem in fermionic Fock space. Precisely, we define the $\cN=2$ supersymmetric version of the local Hamiltonian problem as follows:\\

\mypromprob{{\sc SUSY $k$-local Hamiltonian}}{A local fermionic supercharge, $\cQ=\sum_{a} \cQ_{a}$, acting on the Fock space of $n$ fermionic modes, and two numbers $b'>a'\geq 0$, with  $b'-a'>\frac{1}{\text{poly}(n)}$.}{The  lowest energy level $E_{0}$ of the supersymmetric $k$-local Hamiltonian, $H=\{\cQ,\cQ^{\dagger}\}$, is such that either $E_{0}\leq a'$ or $E_{0}\geq b'$.}{Output $\yes$ if $E_{0}\leq a'$ and $\no$ if $E_{0}\geq b'$. \\}

\noindent Note that although a local supercharge is given as input, the degree of locality $k$ refers to that of the Hamiltonian, which is (at most) twice the degree of locality of $\cQ$. 

To provide a description of a $k$-local supercharge  $\cQ=\sum_{a} \cQ_{a}$ as input it is sufficient to explicitly provide the individual operators $\cQ_{a}$. We may at times drop the restriction that an operator is $k$-local, and instead require that it simply be \emph{sparse}, i.e., it has at most $\poly(n)$ non-zero entries per row and column, and that the input is provided via unitaries that allow us to access the non-zero entries. We call such access \emph{sparse access}.\footnote{We define this as having access to classical (or quantum) subroutines that, given an index $i$, can output a list running over all non-zero entries in row or column $i$, and a unitary that allows us to query the value of the $i,j$th entry of $d$. More precisely, we assume we have access to three unitaries $O_{\text{row}}$, $O_{\text{col}}$, and $O_{d}$, where: $O_{\text{row}} : \ket{i,l}\ket{0} \mapsto \ket{i,l}\ket{I(i,l)}$ and $O_{\text{col}} : \ket{j,l}\ket{0} \mapsto \ket{j,l}\ket{J(j,l)},$ for $0 \leq i, j \leq 2^n$, $0 \leq l \leq \poly(n)$, and where $I(i,l)$ denotes the $l$th non-zero entry in row $i$ of $d$ and $J(j,l)$ the $l$th non-zero entry in column $j$ of $d$. And finally $O_{d} : \ket{i,j}\ket{0} \mapsto \ket{i,j}\ket{d_{ij}}$.} In either case we always assume that the entries of $\cQ$ are upper bounded by a polynomial, and hence that $\|\cQ\| = O(\poly(n))$.

\subsection{The SUSY local Hamiltonian problem is $\QMA$-complete }
\label{sec:The SUSY local Hamiltonian problem is QMA-complete}

 The goal of this section is to prove the following Theorem:
\begin{thm}[]\label{thmQMA}
 {\sc SUSY $k$-local Hamiltonian} is $\QMA$-complete for $k\geq 4$.
\end{thm}

\noindent We begin by  first establishing the following proposition:

\begin{proposition}[]\label{thmQMAcontainment}
 {\sc SUSY $k$-local Hamiltonian} is in $\QMA$ for any constant $k$, and remains in $\QMA$ when the $k$-local restriction is dropped and we require only that $\cQ$ is sparse.
\end{proposition}
\begin{proof}
It suffices to prove the containment for sparse $\cQ$ since any $k$-local operator is automatically sparse. Since a supersymmetric Hamiltonian is a special case of a general Hamiltonian, one may think that {\sc SUSY $k$-local Hamiltonian} is automatically inside $\QMA$. This is not so straightforward since in the supersymmetric version we are given efficient (sparse) access to $\cQ$, which does not in general ensure efficient access to $\{\cQ,\cQ^{\dagger}\}$, which might not be sparse, and hence using Hamiltonian simulation to obtain an approximation of $e^{i\{\cQ,\cQ^{\dagger}\}}$ will generally not be efficient. This difficulty is easily overcome, however,  by noting that the $\cN=2$ SUSY algebra \eqref{nilpQ} and \eqref{QQcom} imply
\equ{
\(\cQ+\cQ^{\dagger}\)^{2}=\{\cQ,\cQ^{\dagger}\}=H\,,
}
and thus the eigenvalues of $H$  are the squares of the eigenvalues of the Hermitian operator $\cQ+\cQ^{\dagger}$, to which we do have efficient access. Thus, a witness for a $\yes$ instance is an eigenstate of the supercharge $\cQ+\cQ^{\dagger}$ with eigenvalue  $-\sqrt{a'}\leq \lambda\leq \sqrt{a'}$, which implies that the smallest eigenvalue of $H$ satisfies $E_{0}\leq a'$. Given such a witness we simply run phase estimation on the unitary operator $e^{i \left(\cQ+\cQ^{\dagger}\right)}$,  or an approximation of it obtained via Hamiltonian simulation using, e.g., the algorithm of~\cite{berry2015hamiltonian}, with the witness as input. Since we are promised that either $E_0 \leq a'$ or $E_0 \geq 
b'$ with $b' - a' > 1/\poly(n)$, we only need to obtain an estimate of the eigenvalue of $\cQ+\cQ^{\dagger}$ on the witness to additive accuracy $\epsilon = 1/O(\poly(n))$ to verify that $E_0 \leq a'$. Using standard techniques, the run-time of Hamiltonian simulation plus phase estimation will be $\poly(n, 1/\epsilon, \|\cQ\|)$, which is $\poly(n)$ for our choice of $\epsilon$, and thus the problem is in $\QMA$ for sparse $\cQ$.
\end{proof}

\noindent As mentioned, we will prove $\QMA$-hardness of the SUSY Hamiltonian problem by reduction from  standard {\sc  $k$-local Hamiltonian}. Since the latter is formulated in qubit space, we will need to provide an encoding of qubit space into fermionic Fock space. There are various ways of doing this, but the most useful to our purposes is to represent a qubit as a single fermion which can be in two different modes, $a_{i},b_{i}$, as in \cite{Liu_2007}. 
Namely, we consider a set of $m=2n$ fermionic modes satisfying the algebra
\equ{\label{basisfN}
\{a_{i}^{\dagger},a_{j}\}=\delta_{ij}\,,\qquad \{b_{i}^{\dagger},b_{j}\}=\delta_{ij}\,,
}
and all other anticommutators vanishing. An $n$-qubit basis state $\ket{n_{1}\cdots n_{n}}$ is mapped into a state $ \ket{\psi_{n_{1}\cdots n_{n}}}$ in fermionic Fock space  as
\equ{\label{basisfN}
\ket{n_{1}\cdots n_{n}}\quad \to \quad \ket{\psi_{n_{1}\cdots n_{n}}}:= (a_{1}^{\dagger})^{1-n_{1}}(b_{1}^{\dagger})^{n_{1}}\cdots  (a_{n}^{\dagger})^{1-n_{n}}(b_{n}^{\dagger})^{n_{n}}\ket{\mathit{vac}}
}
and $\ket{\mathit{vac}}$ is the vacuum state with no fermions. Importantly, note that the states $ \ket{\psi_{n_{1}\cdots n_{n}}}$ have all fermion number $F=n$. 

Single qubit operators are mapped into bilinears in these fermion operators, as
\equ{\label{mapsigmaab}
\sigma_{i}^{+}\to a_{i}^{\dagger} b_{i}\,,\qquad \sigma_{i}^{-}\to b_{i}^{\dagger} a_{i}\,,\qquad  \sigma_{i}^{z}\to  a_{i}^{\dagger} a_{i}-b_{i}^{\dagger} b_{i}\,;
} 
one can easily check this is consistent with the algebra of Pauli matrices.\footnote{This is simply the well known fermionic representation of spin-$\tfrac12$ operators, with $a^{\dagger}=a^{\dagger}_{\uparrow}$ creating a spin up and $b^{\dagger}=a^{\dagger}_{\downarrow}$ creating a spin down.} Note that the map is local, sending qubit operators at site $i$ to bilinears in fermionic operators at site $i$. The total fermion number operator is  
\equ{
F= \sum_{i}(a_{i}^{\dagger} a_{i}+b_{i}^{\dagger} b_{i})\,,
} 
and all bilinears in \eqref{mapsigmaab} have $F=0$. 

Then, given a $k$-local qubit operator $A$
\equ{
 A=\sum_{S}A_S(\sigma_i^{+},\sigma_i^{-},\sigma_i^{z})\,,
}
applying the map \eqref{mapsigmaab} to each term gives 
\equ{
\hat A(a_i,b_i,a_i^{\dagger},b_i^{\dagger}):= \sum_{S}A_S(a_i^{\dagger}b_i,b_i^{\dagger}a_i,a_i^{\dagger}a_i-b_i^{\dagger}b_i)\,,
}
which is a $k$-local operator in fermionic Fock space, with fermion number $F=0$.

Note that $\hat A$ acts on the entire $2^{2n}$-dimensional Fock space. When restricted to states in the space spanned by \eqref{basisfN}, however, it coincides with the qubit operator $A$. Indeed, consider two arbitrary $n$-qubit states, 
\equ{
\ket{\psi}:=\sum c_{n_{1}\cdots n_{n}}\ket{n_{1},\cdots,n_{n}}\,,\qquad \ket{\psi'}:=\sum c'_{n_{1}\cdots n_{n}}\ket{n_{1},\cdots,n_{n}}\,.
}
Applying the map \eqref{basisfN} to each basis element gives the corresponding states in fermionic Fock space,
\equ{
\ket{\psi_n}=\sum c_{n_{1}\cdots n_{n}}\ket{\psi_{n_{1}\cdots n_{n}}}\,,\qquad \ket{\psi_n'}=\sum c'_{n_{1}\cdots n_{n}}\ket{\psi_{n_{1}\cdots n_{n}}}\,.
}
Then, one can see that 
\equ{\label{eq:same_op}
\braket{\psi'|A|\psi}= \braket{\psi'_n|\hat A|\psi_n}\,.
}
This will play an important role in the reduction from hard problems in qubit space to hard problems in fermionic Fock space.  

\begin{thm}[]\label{thmQMAhard}
 {\sc SUSY $4$-local Hamiltonian} is $\QMA$-hard. 
\end{thm}

\begin{proof} 

We will prove hardness by reduction from the standard  {\sc $k$-local Hamiltonian}.  
Consider an instance of  {\sc $k$-local Hamiltonian}  specified by a Hermitian matrix $A=A^{\dagger}=\sum_{S=1}^{L}A_{S}$ acting on $n$ qubits, and parameters $(a,b)$.\footnote{We reserve the symbol $H$ for the actual Hamiltonian of our system which will be related, but not identical, to $A$.} Then, applying the map \eqref{mapsigmaab} we construct the operator
\equ{
B(a_i,b_i,a_i^{\dagger},b_i^{\dagger}):= \sum_{S=1}^{L}A_{S}\(a_i^{\dagger}b_i,b_i^{\dagger}a_i,a_i^{\dagger}a_i-b_i^{\dagger}b_i\)+B_{\mathit{pen}}\,,
}
where
\equ{
B_{\mathit{pen}}= J(n)\sum_{i=1}^{n} \left[n_{a_i} n_{b_i}+(n_{a_i}-1)( n_{b_i}-1)\right]\,,
}
with $J(n)$ a coupling constant. Note that $B$ has fermion number $F=0$. The term $B_{\mathit{pen}}$ penalizes all states which do not have exactly one fermionic mode at each site $i=(1,\dots,n)$. Note that the construction of the operator $B$ requires only a polynomial number of operations since  $L\leq {n \choose k} = O(n^k)$. Let us introduce an auxiliary pair of fermionic modes $(a_{0},b_{0})$ at a site $i=0$. The total set of fermionic operators is  $(a_{0},b_{0},a_{i},b_{i})$ with $i=1,\cdots,n$. Let us consider the supercharge 
\equ{
\cQ= \frac{1}{\sqrt 2}(a_{0}^{\dagger} +b_{0}^{\dagger} )\, B(a_{i},b_{i},a_{i}^{\dagger},b_{i}^{\dagger})\,.
}
Note that the fermionic operators at site $i=0$ do not appear inside $B$; their role is simply to ensure that $\cQ^{2}=(\cQ^{\dagger})^{2}=0$ for arbitrary $B$ and that $\cQ$ has fermion number $F=1$. The hardness of the problem is encoded entirely in $B$. Indeed, the resulting Hamiltonian is 
\equ{\label{SUSYHB2} 
H=\{\cQ,\cQ^{\dagger}\}= B^{2}\,,
} 
and the eigenvalues of $H$ are simply the square of those of $B$. Now, by choosing $J(n)$ large enough (in particular larger than $\|A\| = O(\poly(n))$) the ground space of $B$ is necessarily in the subspace in which $B_{\mathit pen}=0$. Note that 
\equ{
B_{\mathit pen}=0\qquad \Leftrightarrow \qquad (n_{a_i},n_{b_i})=(1,0) \; \text{or} \; (0,1)\,,
}
and  each site $i=(1,\dots,n)$ is occupied by exactly one fermion in this subspace. This is precisely the Fermionic Fock subspace spanned by the states on the RHS of \eqref{basisfN}. The operator $B$ on this subspace coincides with the qubit operator $A$ by \eqref{eq:same_op}.  Thus, if $A$ is a $k$-local Hamiltonian  satisfying the promise with parameters $(a,b)$ then the SUSY Hamiltonian is $2k$-local and satisfies the promise with parameters $(a',b')=(a^{2},b^{2})$. Furthermore, note that
\equ{
b'-a'=(b-a)(b+a)\geq (b-a)^{2}> \frac{1}{\poly(n)}\,,
}
and the separation remains inverse polynomial, as required. Thus, a $\yes/\no$ instance of {\sc $k$-local Hamiltonian} implies a $\yes/\no$ instance of {\sc SUSY $2k$-local  Hamiltonian}. Since  {\sc $2$-local Hamiltonian} is $\QMA$-hard \cite{2004quant.ph..6180K}, it follows that {\sc SUSY $4$-local Hamiltonian} is $\QMA$-hard. 
\end{proof} 

It would be interesting to determine the hardness of  {\sc SUSY $k$-local Hamiltonian} for $2\leq k\leq3$. It would also be interesting to study whether the problem remains hard for systems with a larger degree of supersymmetry, i.e.,  $\cN>2$. Finally, note that although in this construction the supersymmetric Hamiltonian is local, in the sense that it is given by a sum of terms involving a constant number of sites, it is not {\it geometrically} local, even if $B$ is geometrically local since taking a square of a geometrically local function is in general not geometrically local. It would be interesting to study whether {\sc SUSY $k$-local Hamiltonian} remains hard for geometrically local Hamiltonians.

\section{The $k$-local Cohomology Problem}
\label{sec:The k-local cohomology problem}

In this section, we introduce and study the computational complexity of the cohomology problem. As mentioned in the Introduction, this is a fundamental problem in homological algebra, encompassing a large number of problems from a myriad of mathematical fields including  topology, group theory, and number theory. See \cite{CartanEilenberg+2016,weibel_1994} for standard textbooks on homological algebra and \cite{KaczynskiTomasz2004Ch/T} for a more accessible introduction discussing various applications and computational aspects.\footnote{For applications in algebraic topology see, e.g., \url{http://pi.math.cornell.edu/~hatcher/AT/ATplain.pdf} and \cite{CompTopBook,zomorodian_2005} for computational aspects.} 
We will restrict ourselves to a subclass of cohomology problems to which the tools of quantum Hamiltonian complexity can be readily applied but which at the same time is rich enough to encompass a wide range of interesting problems. Precisely,  we introduce the problem {\sc $k$-local cohomology} and show that it is $\QMA_{1}$-hard and contained in   $\QMA$. The setup is  heavily motivated by the supersymmetric systems studied in Section~\ref{sec:The SUSY local Hamiltonian problem is QMA-complete} and captures the computational problems that arise in these systems.

Although we have briefly discussed cohomology in Section~\ref{sec:Cohomology, homology, and ground states} in the context of supersymmetric quantum mechanics, we now review the general definition of a cochain complex and cohomology groups in a self contained way. 

\begin{definition}[Cochain complex]
 A cochain complex is a sequence of Abelian groups, $\{C^{p}\}$, together with a ``coboundary'' operator $d$:
\equ{\label{complexN}
C: \qquad \cdots  \rightarrow  C^{p-1}\xrightarrow{d}C^{p}\xrightarrow{d} C^{p+1} \rightarrow \cdots  \,,
}
with 
\equ{\label{nilpd}
d^{2}=0\,.
}
The elements of $C^{p}$ are called {\it $p$-cochains} or simply cochains.   
\end{definition}
A special role is played by the notion of ``cocycles'' and ``coboundaries.'' 
\begin{definition}[Cocycle, coboundary]
A cochain $c^p\in C^p$ is said to be a {\it $p$-cocycle} if it is in the kernel of the coboundary operator, i.e.,
\equ{
dc^p=0\,.
}
A cochain is said to be a {\it $p$-coboundary} if it is in the image of $d$, i.e., if there exists some $c^{p-1} \in C^{p-1}$ such that
\equ{
c^p=dc^{p-1}\,.
}  
\end{definition}

Note that due to the nilpotency property \eqref{nilpd} if $c^p$ is a cocycle,  so is $\tilde c^p=c^p+dc^{p-1}$. One then says that the cocycles $c^p$ and $\tilde c^p$ are in the same ``cohomology'' class,
\equ{
c^p\sim \tilde c^p\qquad \text{if}\qquad c^p=\tilde c^p+dc^{p-1}\,.
}
The element $c^p$ (or equivalently $\tilde c^p$) is, by definition, a representative of the cohomology class $H^p(d)$. More precisely,
\begin{definition}[Cohomology group]
The $p$th cohomology group is
\equ{
H^{p}(d):= \frac{\text{Ker}(d): C^{p}\to C^{p+1}}{\text{Im} (d): C^{p-1}\to C^{p}}\,,
}
consisting of $p$-cocycles that are not $p$-coboundaries.
\end{definition}
The collection of $p$-cocycles form the subgroup $Z^p$, and the collection of $p$-coboundaries the subgroup $B^p$, so that $H^p(d) = \frac{Z^p}{B^p}$. The relationships between these groups is illustrated in Figure~\ref{fig:cohomology}.

\begin{figure}
    \centering
    \includegraphics[scale=0.7]{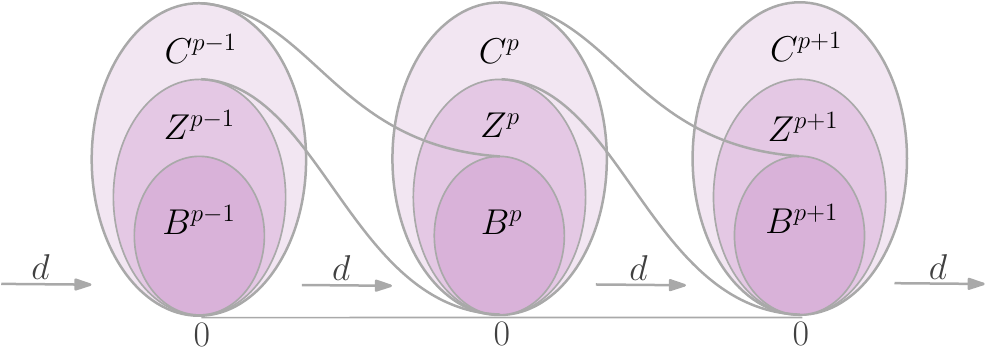}
    \caption{Relationship between the cochains $C^{p-1}, C^p, C^{p+1}$, the cocycles $Z^{p-1}, Z^p, Z^{p+1}$, and the coboundaries $B^{p-1},B^p,B^{p+1}$. The $p$th cohomology group consists of those elements which are cocycles but not coboundaries, i.e. $H^p(d) = Z^p/B^p$.}
    \label{fig:cohomology}
\end{figure}
\begin{definition}[Betti number]
The $p$th Betti number is defined as the rank of the $p$th cohomology group,
\equ{
\beta_{p}:=\text{rank}\, H^{p}(d)\,.
}
\end{definition}
Note that the full cohomology group is  
\equ{
H(d):= \bigotimes_{p}  H^p(d)\,,
}
with total rank the sum of Betti numbers,
\equ{
\text{rank}\, H(d) = \sum_p \beta_p\,.
}
The ``Euler characteristic'' of the complex is defined  as the alternating sum, 
\equ{
\chi:=\sum_p (-1)^p\beta_p\,.
}

Finally, one can similarly define a \emph{chain} complex as a sequence of Abelian groups $\{C_p\}$ connected by a boundary operator $\partial$, satisfying $\partial^2=0$, in which the arrows are reversed compared to cohomology:
\equ{
C: \qquad \cdots  \leftarrow  C_{p-1}\xleftarrow{\partial}C_{p}\xleftarrow{\partial} C_{p+1} \leftarrow \cdots  \,,
}
and one defines the {\it homology} groups
\equ{
H_{p}(d):= \frac{\text{Ker}(\partial): C_{p}\to C_{p-1}}{\text{Im} (\partial): C_{p+1}\to C_{p}}\,,
}
consisting of $p$-cycles that are not $p$-boundaries.

Cohomology is often presented as the ``dual'' version of homology. In our setting cohomology groups will also define homology groups by conjugation of the operator $d$. Since the notion of cohomology connects more naturally to supersymmetry, we thus focus on cohomology from now on. The hardness results we derive for cohomology hold also for homology. 

\

The computational problem we wish to consider is: given a cochain complex $C$, to decide whether a certain cohomology group $H^l(d)$ is nontrivial. We refer to this as the ``cohomology problem:''\\

 \par\noindent  
   {\bfseries Input}: A cochain complex $C$  and an integer $l$.\\
   {\bfseries Problem}: Output $\yes$ if $H^{l}(d)\neq 0$ and $\no$ otherwise.\\
   \par

\noindent Note that a general complex $C$ can consist of an infinite sequence. Since we consider $C$ to be part of the input, we will focus on finite sequences. Furthermore, in order to make contact with quantum Hamiltonian complexity, in which the operator $d$ is taken to act on a Hilbert space, we consider the special case in which the Abelian groups $C^{p}$ have the additional structure of a vector space, which we denote by $V^{p}$. All in all, we consider finite sequences of {\it vector spaces},
\equ{\label{cochainVp}
C:\qquad 0\xrightarrow{d} V^{0}\xrightarrow{d} \cdots \xrightarrow{d} V^{m}\xrightarrow{d} 0\,,
}
and take $m$ to be the length of the input. We will further endow the full vector space $V:=\bigoplus_{p=0}^{m}V^{p}$ with an inner product and a corresponding notion of conjugation, which lets us define the operator $d^{\dagger}$. As we discuss below, this is crucial in our setting for defining the appropriate promise version of the cohomology problem.   From now on we assume that $V$ is a finite dimensional vector space and thus $d$ and $d^\dagger$ are finite dimensional operators acting on this space.

Given a notion of conjugation, it is convenient to introduce the ``Laplacian'' operator
\begin{definition}[Laplacian operator]
For coboundary operator $d$, the Laplacian operator is defined as
\equ{
\Delta:= dd^{\dagger}+d^{\dagger}d\,.
}
\end{definition}

\noindent Since this is a positive semidefinite operator, it follows that 
\equ{
\Delta v =0 \qquad \Leftrightarrow \qquad  dv  =d^\dagger v=0\,.
}
It is also easy to see that an element $v\in V$ is an element of cohomology $H(d)$ if and only if it is a zero-eigenvalue of the Laplacian:
\equ{\label{Hodge}
v\in H(d) \qquad \Leftrightarrow \qquad \Delta v=0\,.
}
In the context of de Rham cohomology this is known as Hodge's theorem. In the language of supersymmetry, this is the statement that elements of cohomology of the coboundary operator are in one-to-one correspondence with supersymmetric ground states and we have given a proof of this below Eq.~\eqref{qomega0}.

Let $d^l:=d|_{V^l}$ and ($d^\dagger)^l:=d^\dagger|_{V^l}$ denote the restriction of the operators $d$ and $d^\dagger$ to the space $V^l$, respectively.\footnote{We will often omit the subscript on $d$ when this is redundant to simplify notation. } Similarly, let $\Delta^l:= \Delta|_{V^l}$ denote the restriction of the Laplacian to $V^l$. Then, one can write\footnote{ Note that $(d^\dagger)^{l+1}=(d^{l})^\dagger$ and thus one can also write $\Delta^l=d^{l-1} (d^\dagger)^l+(d^\dagger)^{l+1} d^l$.} 
\equ{
\Delta^l = d^{l-1} (d^\dagger)^{l-1}+(d^\dagger)^{l} d^l \,.
}
It follows trivially from  \eqref{Hodge} that  $v^l\in V^l$ is an element of the cohomology group $H^l(d)$  if it is a zero-eigenvalue of the corresponding Laplacian: 
\equ{\label{Hodgel}
v^l\in H^l(d) \qquad \Leftrightarrow \qquad \Delta^l v^l=0\,.
}

\  

Although the cohomology problem is rather fundamental, and it is believed to be hard, little is known about its precise complexity. In the case of simplicial complexes, the question was posed  in \cite{2002math......2204K}. It was recently shown in \cite{ADAMASZEK20168} that for  independence complex the problem is $\NP$-hard (given a graph as input) and generally not in $\NP$.\footnote{The problem is described in \cite{ADAMASZEK20168} in terms of homology rather than cohomology, but as already mentioned these are equivalent for these complexes. } As we describe below, however, a certain promise version of the cohomology problem can be placed inside $\QMA$. 

\subsection{The $k$-local cohomology problem}

We now introduce a version of the cohomology problem, inspired by supersymmetric systems discussed in previous sections. Indeed, note that the supersymmetric systems discussed in Section~\ref{sec:Cohomology, homology, and ground states} precisely lead to a cochain structure of the form \eqref{cochainVp}, with the  coboundary operator identified with the supercharge and the Laplacian with the supersymmetric Hamiltonian: 
\equ{
\cQ=d\,,\qquad  H=\Delta\,.
}
In fact, this is how the notion of cohomology was introduced there. The supersymmetric models of Section~\ref{sec:Supersymmetry in lattice systems} are of particular interest since they describe systems with discrete degrees of freedom and are thus more amenable to a computational treatment.\footnote{It would also be interesting to study appropriate dicretizations of continuum quantum mechanical systems. } We thus take inspiration from these to  define the cohomology problem we will study. First, we note that in these models the vector space $V$ is the Hilbert space of the system, which can generally be a {\it subset} of the fermionic Fock space of $m$ fermions,
\equ{
V\subseteq \text{span}\{\ket{n_{1},\cdots,n_{m}}\}\,,\qquad n_i=\{0,1\}\,.
}
This happens whenever certain states are excluded from the Hilbert space for dynamical reasons, such as a strong repulsion forbidding certain states corresponding to nearby particles. The grading of $V=\bigoplus_{p=0}^m V^p$ is  provided by the fermion number $F$, i.e., 
\equ{
V^p=\text{span}\{v\in V\;|\; F v=p \, v\}\,.
}
We take the coboundary operator to have fermion number 1 and thus it acts as $d: V^p\to V^{p+1}$.  Given such a cochain complex $C=(V,d)$ and an integer $l$ we would like to determine whether the cohomology group $H^l(d)$ is nontrivial. We will input the description of a cochain complex as a tuple $(V[m],d[k])$, where for clarity we use the notation $[\cdot]$ to explicitly denote a parameter of $V$ or $d$. As always, we need to provide a succinct description of the input. For our purposes, we will take  (unless stated otherwise) the input to satisfy the following:
\begin{itemize}
    \item The graded vector space $V[m] = V = \bigoplus_{p=0}^m V^p$ is a vector space spanned by a {\it subset} of basis states in the fermionic Fock space of $m$ fermionic modes, $\ket{n_{1},\cdots,n_{m}}$. The subset will be determined by a polynomial-sized (in $m$) list of constraints among the $n_{i}$. This list might be specified in a number of ways: for instance, we could be given a list of constraints via a (polynomial) set of Boolean functions, a Boolean formula in conjunctive normal form, or an arithmetic circuit over the integers. We will mostly remain agnostic about how these constraints are supplied, only that they contain enough information to define the full vector space $V$. Often we will consider when this space can be efficiently ``sampled'' from (see Definition~\ref{def:efficient_sample}).
    
    \item The coboundary operator $d[k] = d$ is a $k$-local operator acting on $V$ and with fermion number $F=1$. Any such operator can be written as 
\equ{\label{dB}
d=\sum_{i=1}^{m}a_{i}^{\dagger}\, B_{i}(a,a^{\dagger})\,,
}
where the $B_{i}$ are $(k-1)$-local functions  such that $d^{2}=0$. Furthermore, each $B_i$ has $F=0$ which ensures that  $d:V^{p}\to V^{p+1}$ and defines a cochain of the form \eqref{cochainVp}.\footnote{One could relax this condition and simply require that the $B_{i}$ have even fermion number. This defines a $\Bbb Z_{2}$-graded complex. See Section~\ref{sec:Cohomology, homology, and ground states} for a discussion in the context of supersymmetric systems. Since a $\Bbb Z_{2}$-graded complex is obtained from the complexes we study, identifying $F$ mod 2, all the hardness results we derive here hold for $\Bbb Z_{2}$-graded complexes as well. } Note that since $V$ is a subspace of fermionic Fock space, there is a natural notion of conjugation given by the standard Hermitian conjugate. Indeed, we define
\equ{\label{dBd}
d^{\dagger}=\sum_{i=1}^{m} B_{i}^{\dagger}(a,a^{\dagger})\, a_{i}\,.
}
\end{itemize}
The size of the input in this case is $m$, the number of fermionic modes. This is natural, since the full Fock space on $m$ modes will have dimension $2^m$, and the topological objects we consider (e.g. Laplacians) will often be $\approx 2^m$-dimensional.

\

Let us note that if, instead,  the {\it full} vector space $V$ is provided as a set of all its basis elements, and the coboundary operator as the full matrix acting on these basis elements, the cohomology problem is in $\P$ (taking  the dimension of $V$ to be the size of the input). Indeed, given such a description one simply needs to diagonalize the Laplacian operator $\Delta$ and check, for a given $l$, whether there are any zero-eigenvalues in the subspace $V^{l}$. If so, the cohomology group $H^{l}(d)$ is nontrivial and trivial otherwise. Of course, the representation of $V$ as a list of all its basis elements is not efficient, since it can contain an exponential (in $m$) number of elements.

\ 

We comment that the description of the complex $(V[m],d[k])$ given above is a natural generalization of some instances considered in the literature. In particular, the problem of determining whether the $l$th {\it homology} group $H_l(\partial)$ of the independence complex was recently studied in  \cite{ADAMASZEK20168}, where the input to the problem is taken to be a graph $G$. In our setting, the input $(V[n],d[k]) $ is specified as follows.\footnote{The homology groups $H_l(\partial)$ with $\partial$ the standard boundary operator of simplicial complexes are isomorphic to the cohomology groups of the coboundary operator in \eqref{dindepend}, $H_l(
\partial)\simeq H^l(d)$ so these are equivalent computational problems--we expand on this in Section~\ref{sec:Betti numbers in Topological Data Analysis}. }  The vector space is defined by the $\cO(n^2)$ number of constraints
\equ{
A_{ij}n_i n_j=0\,,
}
where $A_{ij}$ is the incidence matrix of $G$. All solutions to these constraints correspond to independent sets of $G$. The coboundary operator is given by
\equ{\label{dindepend}
d[k]=\sum_i a_i^\dagger P_i\,,\qquad P_i=\prod_{j} (1-A_{ij}n_j)\,,
}
where $k=\delta+1$  with $\delta$ the maximum degree of the graph. Thus both $V[n]$ and $d[k]$ can be computed efficiently from $G$. Our description of chain complexes is more general and thus encompasses a larger number of problems.  

\ 

Having discussed the form of the input, we define the following problem:\\
\mypromprob{{\sc $k$-local Cohomology}}{A cochain complex $(V[m],d[k])$ and an integer $l$.}{Either there exists a state $\ket{\Omega^{l}}\in V^{l}$ such that $d\ket{\Omega^{l}}=d^{\dagger}\ket{\Omega^{l}}=0$, and hence $H^l(d)\neq 0$, or otherwise $\norm{\(d\pm d^{\dagger}\)\ket{\Psi^{l}}}\geq \epsilon$ for all $\ket{\Psi^{l}}\in V^{l}$, with  $\epsilon=\frac{1}{\poly(m)}$.}{Output $\yes$ if the former and $\no$ if the latter.\\}

\noindent Note that the promise requires the notion of conjugation and an inner product, which is why the groups $C^{p}$ are taken to be vector spaces endowed with an inner product and notion of conjugation.  

\ 

One can similarly define the ``dual'' version of the problem, specified instead by a vector space $V=\bigoplus_{p=0}^m V_p$ and a  $k$-local {\it boundary} operator $\partial: V_p\to V_{p-1}$, $\partial^2=0$, in the same way:\\

\mypromprob{{\sc $k$-local Homology}}{A chain complex $(V[m],\partial[k])$ and an integer $l$.}{Either there exists a state $\ket{\Omega_{l}}\in V_{l}$ such that $\partial\ket{\Omega_{l}}=\partial^{\dagger}\ket{\Omega^{l}}=0$,  and hence $H_l(\partial)\neq 0$, or otherwise $\norm{\(\partial\pm \partial^{\dagger}\)\ket{\Psi_{l}}}\geq \epsilon$ for all $\ket{\Psi_{l}}\in V_{l}$, with  $\epsilon=\frac{1}{\poly(m)}$.}{Output $\yes$ if the former and $\no$ if the latter.\\}

\noindent Note that since the notion of conjugation is defined for all these complexes, the operator $d^{\dagger}$ defines a boundary operator $\partial$ and it follows that all the complexity results we derive for  {\sc $k$-local Cohomology} also hold for  {\sc $k$-local Homology}. We will not make this explicit below.

\

We note that  {\sc $k$-local Cohomology} (and {\sc $k$-local Homology}) admits  an efficient  algorithm for {\it constant} $l$, even if $V$ itself is exponential in size. Indeed, for constant $l$ there are at most $O(m^{l})$ elements in $V^{l}$. Thus, the operators $d^l$ and $\Delta^l$ are polynomial-size matrices and determining if they have any zero-eigenvalues can be done in polynomial time (see \cite{CompTopBook} for an overview of various algorithms). 

For non-constant $l$, however, the matrices can be exponentially large and we expect no  algorithm (classical or quantum) to be efficient. Indeed, strong evidence for this is provided by the fact that these problems are $\QMA_1$-hard as we show below. We will prove hardness by reduction from {\sc Quantum $k$-$\SAT$}, introduced in \cite{bravyi2011efficient}. For reference, recall this is defined as:\\
\mypromprob{{\sc Quantum $k$-$\SAT$} \cite{bravyi2011efficient}}{A set of $k$-local projectors $\Pi_S$ where $S$ are possible subsets of $\{1,\ldots,n\}$ of cardinality $k$.  }{Either there exists a state $\ket{\psi}$ such that $\Pi_S\ket{\psi}=0$ for all $S$ or otherwise $\sum_{S}\bra{\psi}\Pi_{S}\ket{\psi}\geq \epsilon$ for all $\ket{\psi}$, with $\epsilon>\frac{1}{\poly(n)}$. }{Output $\yes$ if the former and $\no$ if the latter.\\}

\noindent It was shown in \cite{bravyi2011efficient} that this problem is $\QMA_1$ complete for $k=4$ and that for $k=2$ there is an efficient algorithm. For  $k=3$ the problem was shown to be $\QMA_1$-hard in \cite{Gosset_2013}.

\subsection{The complexity of the $k$-local cohomology problem}

We now prove the following Theorem:

\begin{thm}[]\label{thmCohomology}
{\sc $4$-local Cohomology} is $\QMA_{1}$-hard.
\end{thm}

\begin{proof}

We will establish hardness by reduction from  {\sc Quantum $k$-$\SAT$}. That is, our goal is to construct a cochain complex $(V,d)$ such that a certain cohomology group $H^l(d)$ is nontrivial only if there exists a satisfying assignment of {\sc Quantum $k$-$\SAT$}.  The general structure of the reduction is identical to that used in the proof of Theorem~\ref{thmQMAhard}.\footnote{In fact, one could establish hardness of the cohomology problem by choosing $B$ there to be Kitaev's clock Hamiltonian \cite{KitaevBook} and set the perfect completeness parameter $c=1$. However, this would lead to a $6$-local supercharge and thus the hardness result is strengthened by reduction from  {\sc Quantum $k$-$\SAT$}.} 

Consider an instance of {\sc Quantum $k$-$\SAT$}, given by a set of $k$-local projectors $\Pi_S$ acting on an $n$-qubit state  $\ket{\psi}$. These have the form
\equ{
\Pi_{S}=\Pi_{S}(\sigma_{i}^{+},\sigma_{i}^{-},\sigma_{i}^{z})\,,\qquad \ket{\psi}=\sum c_{n_{1}\cdots n_{n}}\ket{n_{1},\cdots, n_{n}}\,,
}
where the arguments in $\Pi_S$ denote the qubits $i\in S\subset \{1,\cdots,n\}$ that are acted upon by each term and the $c_{n_{1}\cdots n_{n}}$ are arbitrary coefficients. Based on this data and applying the map \eqref{basisfN},   \eqref{mapsigmaab}  we construct 
\equ{\label{PsPsi}
\cP_{S}:= \Pi_{S}(a_{i}^{\dagger}b_{i},b_{i}^{\dagger}a_{i},a_{i}^{\dagger}a_{i}-b_{i}^{\dagger}b_{i})\,,\qquad \ket{\psi_{n}}:=\sum c_{n_{1}\cdots n_{n}} \ket{\psi_{n_{1}\cdots n_{n}}}\,.
} 
 Note that the $\cP_{S}$ are $k$-local operators with $F=0$ and $\ket{\psi_n}$ is a state in fermionic Fock space with fermion number $F=n$. 
We  introduce an additional auxiliary site $i=0$ with corresponding fermionic modes $(a_0,b_0)$. Then, we construct the coboundary operator 
\equ{\label{dPs}
d=\frac{1}{\sqrt 2}(a_{0}^{\dagger}+b_{0}^{\dagger}) \, \sum_{S} \cP_{S}\,.
}
 Note that the fermionic modes at the site $i=0$ appear only outside $\cP_{S}$; $l$th is to ensure that $d^{2}=0$ for any $\cP_{S}$ and that $d$ has fermion number $F=1$. 
 
 In principle, the operator \eqref{dPs} acts on the full fermionic Fock space $(\Bbb C^2)^{\otimes m}$, $m=2n$, spanned by the basis states $\ket{n_{a_{0}},n_{b_{0}}\,,\dots, n_{a_{n}},n_{b_{n}}}$, where $n_{a_{i}}=\{0,1\}$ and $n_{b_{i}}=\{0,1\}$ are the occupancy numbers of the two types of fermionic modes at site $i$. However, we can consistently restrict the action of $d$ to a vector space $V\subseteq (\Bbb C^2)^{\otimes m}$, which we take  to be the subspace spanned by the basis states satisfying the $n$ constraints:
\equ{\label{constVhard}
n_{a_{i}}n_{b_{i}}=0\,,\qquad \qquad  i=(1,\cdots,n)\,.
}
Note that there are no constraints on the occupancy numbers $n_{a_{0}}$ and $n_{b_{0}}$ at the site $i=0$. For the remaining $n$ sites the constraints above impose that no more than a single mode is occupied.
The vector space then has the form $V=\bigoplus_{p=0}^{n+2} V^p$, where each $V^{p}$ is spanned by the set of states satisfying \eqref{constVhard} and  fermion number $p$. Note that the coboundary operator \eqref{dPs} creates states at the (unconstrained) site $i=0$, while at the sites $i=(1,\cdots, n)$ it exchanges $a$-modes by $b$-modes preserving the condition \eqref{constVhard}. 
Furthermore, since $d$ has $F=1$ it acts as $d: V^{p}\to V^{p+1}$ and this defines a cochain complex of the form
\equ{\label{chardV}
C:\qquad 0\xrightarrow{d} V^{0}\xrightarrow{d} \cdots \xrightarrow{d} V^{n+2}\xrightarrow{d} 0\,.
}

Having specified the cochain complex $(V,d)$ it remains only to set the integer $l$ specifying the particular cohomology group $H^{l}$ of interest. We set $l=n+2$, the highest possible nontrivial value in the complex \eqref{chardV}. That is, we need to search for cochains in $V^{n+2}$ which are cocycles but not coboundaries. Since the sites $i=(1,\cdots, n)$ can be occupied by at most one fermionic mode, any element  $\ket{\Psi_{n+2}}\in V^{n+2}$ takes the form
\equ{\label{staten+2}
\ket{\Psi_{n+2}}=\ket{1,1}\otimes \ket{\psi_{n}}\,,
}
where $\ket{\psi_{n}}$ is an $n$-fermion state of the form in  \eqref{PsPsi}. 

Note that applying the operator $d$ in  \eqref{dPs} to  a  state of the form  \eqref{staten+2}, we automatically have
\equ{\label{dn2psi}
d\ket{\Psi_{n+2}}=0\,,
}
consistent with the form of the complex \eqref{chardV}.
That is,  all $(n+2)$-cochains are automatically $(n+2)$-cocycles. 
On the other hand, applying $d^\dagger$ on such states we have
\equ{
d^{\dagger}\ket{\Psi_{n+2}}=\frac{1}{\sqrt 2}(\ket{0,1}+\ket{1,0})\otimes \sum_{S}\cP_{S}\ket{\psi_{n}}\,.
}
Now, given a $\yes$ instance of {\sc  quantum $k$-$\SAT$}, there exists some $n$-qubit state  $\ket{\psi}$ such that $\Pi_S \ket{\psi}=0$ for all $S$. The corresponding state $\ket{\psi_n}$ in fermionic Fock space obatained by \eqref{PsPsi} then satisfies $\cP_S \ket{\psi_n}=0$. Thus, 
\equ{
 \Pi_{S}\ket{\psi}=0  \qquad \Rightarrow \qquad d^{\dagger}\ket{\Psi_{n+2}}=0\,.
}
This, together with \eqref{dn2psi} implies $H^{n+2}\neq 0$. Thus, for every $\yes$ instance of {\sc  quantum $k$-$\SAT$} there is $\yes$ instance  of {\sc $(k+1)$-local Cohomology}  and vice versa.

To complete the reduction it remains only to check the case of a $\no$ instance. That is, we evaluate the norm
\eqss{
\norm{(d\pm d^{\dagger})}_{V^{n+2}}=\,& \sqrt{ \bra{\Psi_{n+2}}(d\pm d^{\dagger})^{2} \ket{\Psi_{n+2}}}= \left( \bra{\psi_{n}} \left(\sum_{S}\cP_{S}\right)^{2}\ \ket{\psi_{n}}\right)^{1/2} \\
\geq \,&  \bra{\psi_{n}} \sum_{S}\cP_{S}\ \ket{\psi_{n}}=\sum_{S}\bra{\psi}\Pi_{S}\ket{\psi}\,,
}
where in the penultimate step we used the Cauchy–Schwarz inequality and in the last step we wrote the expression in qubit space, using \eqref{eq:same_op}. In a $\no$ instance of {\sc  quantum $k$-$\SAT$}, the last expression above is bounded below by $\epsilon$, and we have
\equ{
\norm{d\pm d^{\dagger}}_{V^{n+2}}\geq \epsilon\,.
}
Thus, a $\no$ instance of {\sc  quantum $k$-$\SAT$} implies a $\no$ instance of {\sc $k$-local cohomology}. Since {\sc  quantum $3$-$\SAT$} is $\QMA_{1}$-hard \cite{Gosset_2013} then {\sc $4$-local Cohomology} is $\QMA_{1}$-hard.  

\end{proof}

\noindent One can also see that the problem is contained in  $\QMA$, as we show next.

\begin{thm}
    {\sc $k$-local Cohomology} is contained in $\QMA$ for cochain complexes $(V[m],d[k])$.
\end{thm}
\begin{proof}
    Recall that a $\yes$ instance of {\sc $k$-local Cohomology} is one where there exists some state $\ket{\Omega^l} \in V^l$ such that $d\ket{\Omega^l} = d^\dagger\ket{\Omega^l} = 0$, and in a $\no$ instance, we have that for any state $\ket{\psi^l} \in V^l$, $\| (d \pm d^\dagger) \ket{\psi^l}\| \geq \epsilon$, with $\epsilon = 1/\poly(m)$. Thus for a $\yes$ instance the witness is simply the state $\ket{w^l} =\ket{\Omega^l}$. To verify the witness, it suffices to estimate $\braket{w^l|\Delta^l|w^l} = \braket{w^l|\Delta|w^l}$ up to an inverse polynomial additive accuracy smaller than $\epsilon$. This can easily be done by running phase estimation on an approximation of $e^{i\Delta}$ obtained from Hamiltonian simulation, and providing $\ket{w^l}$ as the input. 
    Since $d$ is $k$-local, $\Delta$ is $2k$-local and the local terms can be  computed from $d$ efficiently. One can use, for instance, the Trotterization method of Hamiltonian simulation to implement a unitary $U$ such that $\|U - e^{it\Delta}\| \leq 1/\poly(m)$ in time $\poly(m,t)$~\cite{lloyd1996universal}. Hence, phase estimation can be used to obtain an estimate of $\braket{w|\Delta|w}$ up to additive accuracy $1/\poly(m)$, with probability at least $2/3$, in polynomial time. Hence, the verifier circuit will measure an energy inverse-polynomially close to zero and accept $\ket{w}$ with probability at least $2/3$, as required.  
    
    In a $\no$ instance, we are promised that any candidate witness $\ket{\psi^l} \in V^l$ will have $\braket{\psi^l|\Delta|\psi^l} \geq \epsilon^2$. However, we are not provided any promises on states outside of $V^l$ and it is possible that we could be ``tricked'' into accepting a witness $\ket{w}$ with support outside $V^l$, and so we must first project the witness onto the space $V^l$. This is straightforward: we compute the Hamming weight of the state $\ket{w}$ in the witness register, and store the result in an ancilla register of $\log m$ qubits. We then proceed to measure this register and if we see any outcome $\neq l$, reject. Otherwise, we know that the {\it collapsed} state in the witness register is supported entirely in $V^l$, and we proceed as above, rejecting if we obtain an energy larger than some threshold $\epsilon^2 - 1/\poly(m)$. By the problem promise and the behaviour of phase estimation, we will therefore reject with probability $\geq 2/3$ as required. 
\end{proof}

Therefore we find that, in general, the complexity of {\sc $k$-local Cohomology} is somewhere in between $\QMA_{1}$ and $\QMA$. The question of whether $\QMA_{1}=\QMA$ is an open problem in complexity theory~\cite{aaronson2009perfect}.

An interesting question is whether the cohomology problem remains $\QMA_{1}$-hard for natural instances of the problem. As mentioned in the Introduction, a system of particular interest is the fermion hard-core model of \cite{Fendley_2003}, which gives rise to the well-known independence complex. A lower bound on the computational complexity of the cohomology problem for the independence complex follows from the results in \cite{ADAMASZEK20168}. In this reference, the chain complex was defined by the boundary operator $\partial$ and it was shown that deciding if the corresponding homology groups $H_l(\partial)$ were trivial or nontrivial is $\NP$-hard, giving a lower bound on its complexity. This translates directly into a lower bound on the complexity of the cohomology problem since in this case the homology and cohomology groups are isomorphic, $H_l(\partial)\simeq H^l(d)$. Finally, we note that for the case of the independence complex, the Laplacian is always sparse, even though the coboundary operator is not always local~\cite{gyurik}. Therefore, following the proof above, we can run Hamiltonian simulation directly on $\Delta$ without first obtaining a local decomposition, and hence the cohomology problem for the independence complex is inside $\QMA$. Combining the two results, we find that the complexity of the problem lies somewhere between $\NP$ and $\QMA$. Determining the precise complexity would be interesting, especially in relation to its applications to topological data analysis \cite{LloydetalTDA}.  We discuss this in more detail in Section~\ref{sec:Betti numbers in Topological Data Analysis}.

\section{The Betti Number Problem}\label{sec:betti_numbers}

In this section, we consider the counting version of the cohomology problem. Namely, rather than deciding whether a certain cohomology group $H^l(d)$ is nontrivial, we consider the problem of determining its rank, or Betti number $\beta_l=\text{rank}\, H^l(d)$.\footnote{Note that Betti numbers only count the rank of the freely acting part of the cohomology group. In the absence of torsion coefficients, as we assume here, the Betti number captures the rank of the full group. Thus, a vanishing Betti number implies a trivial cohomology (and dual homology) group.} 
We consider both the exact problem as well as that of finding an  additive approximation to normalized Betti numbers, $\frac{\beta_l
}{\dim V^l}$. It is known that computing Betti numbers can be a very hard problem. For instance, in the case of  algebraic varieties the problem is $\PSPACE$-hard~\cite{scheiblechner2007complexity} and in the case of the independence complex it is $\#\P$-hard (see below). To the best of our knowledge, the complexity of an additive approximation to normalized Betti numbers has not been directly considered. 

As we show below, for a cochain complex $(V[m],d[O(1)])$, the problem of computing the $l$th Betti number is complete for the class $\#\BQP$ (the counting version of $\QMA$), which was shown to be equivalent to $\#\P$ under weakly parsimonious reductions in \cite{Brown_2011}. Let us first point out that a weaker hardness result follows from \cite{Crichigno:2020vue}, where it was shown that computing the Euler characteristic (or Witten index), 
\equ{\cI=\sum_{l=0}^{m} (-1)^l\beta_l\,,
} 
of a complex is $\#\P$-complete, with the vector space and a $6$-local coboundary operator given as the input. It thus follows that computing the individual Betti numbers themselves is $\#\P$-hard under a Turing reduction. This result is strengthened by the following Theorem:
\begin{thm}
    Given a cochain complex $(V[m],d[6])$, the problem of computing the $l$th Betti number is $\sharpP$-hard for $l = \Omega(n)$ and $\sharpP$-complete provided the $l$th Laplacian $\Delta_l$ has a $1/\poly(m)$ gap above its zero groundspace.
\end{thm}
\begin{proof}
   
    In the proof of Theorem~\ref{thmCohomology}, we showed how to construct a cochain complex from an arbitrary $n$-qubit, $k$-local Hamiltonian such that the Betti number $\beta_{n+2}$ is equal to the number of $0$-eigenvalues of that Hamiltonian.

    In~\cite{Brown_2011} the problem {\sc LH$(\sigma)$} of computing the number of states in a low-energy groundspace of energies $0\leq E\leq \sigma$ of a local $n$ qubit Hamiltonian is considered, and shown to be complete for the class $\mathsf{\#BQP}$ for $e^{-\poly(n)}\leq \sigma\leq 1/\poly(n)$. In our case we are interested in computing the dimension of the zero-eigenspace of a local Hamiltonian, i.e., the problem {\sc LH$(0)$}. This can be seen to be hard for the counting version of $\QMA_1$ -- which we will call $\#\BQP_1$ -- using the same proof as~\cite{Brown_2011} (in particular, considering the perfect completeness version of the  proof of their Theorem~11). Note that by definition we have the inclusions $\sharpP \subseteq \#\BQP_1 \subseteq \#\BQP$. Since $\mathsf{\#BQP}=\sharpP$ under weakly parsimonious reductions  \cite{Brown_2011}, it follows that {\sc LH$(0)$} and hence the problem of computing $\beta_{n+2}$ for these complexes is  $\#\BQP$-hard under weakly parsimonious reductions, and in particular that computing Betti numbers is $\sharpP$-hard.

The problem can be placed inside $\#\BQP$ whenever the Laplacian $\Delta_l$  has at least an inverse-polynomial gap above 0. Indeed, since the Betti numbers $\beta_l$ correspond to the dimension of the zero-eigenspace of the $l$th Lapacian, $\Delta_l$, one can construct a $\BQP$ verifier that distinguishes zero from non-zero eigenvectors so long as  $\Delta_l$ has at least an inverse-polynomial gap above 0. The problem of computing Betti numbers then reduces to the problem of counting the number of inputs that cause the $\BQP$ verifier to accept with high probability, which is inside $\#\BQP$ and by the results of \cite{Brown_2011} is inside $\#\P$. In the construction used in the proof of Theorem~\ref{thmCohomology} mentioned above, we do indeed have an inverse-polynomial spectral gap above 0 (this can be seen from Lemma 12 of \cite{Brown_2011} plus the perfect completeness modification mentioned above), and hence it follows that the Betti number problem is $\#\P$-complete. 
\end{proof}

\ 

\noindent We now focus on the problem of additively approximating Betti numbers. More precisely, we focus our attention on the relatively easier problem of estimating a \emph{normalized} version of Betti numbers. Concretely, we define the following problem:

\ 

\myprob{\label{def:bne}{\sc Betti number estimation (BNE)}}{A cochain complex $(V[m],d[k])$, an integer $0\leq l \leq m$, and an accuracy parameter $\epsilon = \frac{1}{O(\poly(m))}$ and success probability $\mu > 1/2$}{Output an approximation $\chi$ s.t. $\left|\chi - \frac{\beta_l}{\dim{V^l}} \right| \leq \epsilon$ with probability $\geq \mu$.\\}

\noindent Intuitively, {\sc BNE} should be much easier to solve than the cohomology problem considered in Section~\ref{sec:The k-local cohomology problem}. Indeed,  we describe a $\BQP$ algorithm for a general family of complexes in Section~\ref{sec:quantum_algorithms}, showing that for a large class of complexes {\sc BNE} is  tractable for quantum computers. In contrast, no efficient classical algorithms are known for estimating normalized Betti numbers (even for ``easier'' complexes like the clique complex), and indeed it seems reasonable to expect that the problem should remain hard for a large family of complexes. 

Pinpointing the precise complexity of {\sc BNE}, however, is a bit subtle. As we have discussed, Betti numbers correspond to the dimension of the zero-eigenspace of the Laplacian operator $\Delta$, and this creates difficulties from ``both directions'' when attempting to prove completeness results. First, distinguishing zero-eigenvalues from non-zero eigenvalues is in general a hard task for quantum computers in the absence of any promise on a spectral gap above 0 (as we had for the problems considered in Section~\ref{sec:The k-local cohomology problem}), which can make it difficult to design efficient quantum algorithms. Second, any constructions used in hardness proofs will need to use information contained only in the number of zero-eigenvectors of a Laplacian, which can, at least using our techniques, make it difficult to prove hardness results for complexity classes that don't have a notion of perfect completeness (it is for this reason that we could only prove $\QMA_1$-hardness in Section~\ref{sec:The k-local cohomology problem}). 

To allow for a more complete characterisation of the problem, we instead consider a relaxed version corresponding to counting the number of small eigenvalues (below an inverse-polynomial threshold) of $\Delta_l$, and argue below that this is in fact a natural problem. To capture the notion that this problem is very similar in nature to estimating Betti numbers, and in fact coincides with it when we are promised that $\Delta_l$ has a spectral gap above its zero groundspace, we call this problem the \emph{quasi-Betti number estimation} ({\sc QBNE}) problem. This problem was also considered for the case of clique complexes in~\cite{gyurik}, where it was called the \emph{Approximate Betti number estimation} ({\sc ABNE}) problem, although its hardness was not considered there. We therefore use the name {\sc QBNE} to distinguish this more general version of the problem from the restriction to clique complexes, where it is called {\sc ABNE}.

First we define the notion of spectral density \emph{within a subspace}. Let $A|_{V}$ be the restriction of a matrix $A$ to its action on the space $V$, then for a positive semi-definite Hermitian matrix $A \in \mathbb{C}^{2^n \times 2^n}$ and a threshold $b \in \mathbb{R^+}$, we define the \emph{spectral density within a subspace} of $A$ with respect to threshold $b$ and space $V$ as 
\equ{
	N_A(b,V) := \frac{1}{\dim(V)} \sum_{j : \lambda_j(A|_{V}) \leq b} 1
}
where $\lambda_j(A|_{V})$ denotes the $j$th eigenvalue of $A|_{V}$. When the space $V$ is taken to be the full $2^n$-dimensional complex vector space, this quantity is just the ordinary low-lying spectral density defined in~\cite{brandao_thesis}, denoted by $N_A(b)$. 

The central problem that we consider in this section is the following:\\
\myprob{\label{def:QBNE}{\sc quasi-Betti number estimation (QBNE)}}{A cochain complex $(V[m],d[k])$, an integer $0\leq l \leq m$, precision parameters $\delta, \epsilon = \Omega(1/\poly(m))$, and a success probability $\mu > 1/2$.}{Output an approximation $\chi$ that satisfies, with probability at least $\mu$, \equ{\label{eq:quasi_betti_number_estimation}N_\Delta(b,V^l) - \epsilon \leq \chi \leq N_\Delta(b+\delta,V^l) + \epsilon,} where $\Delta$ is the Laplacian of the complex.\\}

\noindent This problem is the generalization of the {\sc LLSD} problem~\cite{gyurik} from arbitrary Hermitian matrices to Laplacians of cochain complexes. For completeness, and because we make use of it later on, we define the {\sc LLSD} problem here.\\

\myprob{Low-lying spectral density ({\sc LLSD})~\cite{gyurik}}{A positive-semidefinite Hermitian matrix $A \in \mathbb{C}^{2^n \times 2^n}$, a real number $b = \Omega(1/\poly(n))$, precision parameters $\delta, \epsilon = \Omega(1/\poly(n))$, and a success probability $\mu > 1/2$}{Output an approximation $\chi$ that satisfies, with probability at least $\mu$, \equ{N_A(b) - \epsilon \leq \chi \leq N_A(b+\delta) + \epsilon.}}

\

We argue that the fraction of small eigenvalues of the $l$th Laplacian $\Delta_l$ is in fact a natural quantity to consider. For the ordinary graph Laplacian ($\Delta_0$ of the graph's simplicial complex) it is well known that the number of $0$-eigenvalues (i.e. $\beta_0$) gives the number of connected components of the graph. Cheeger’s inequality and its variants provide an approximate version of this fact for the case of two components: they state that a graph has a sparse cut if and only if there are at least two eigenvalues that are close to zero~\cite{alon1986eigenvalues}. Intuitively, this says that if the graph Laplacian has an ``almost-zero'' eigenvalue, then only a few edges with small total weight need to be removed from the graph to cut it into two connected components. It was long conjectured that an analogous characterization should hold for higher multiplicities: that there are $k$ eigenvalues close to zero if and only if the vertex set can be partitioned into $k$ subsets, each defining a sparse cut. Surprisingly, this conjecture was only confirmed quite recently, where a series of works culminated in the so-called higher-order Cheeger inequality~\cite{lee2014multiway}. These results have been used to give theoretical backing to the use of spectral clustering algorithms for optimally partitioning the vertices of a graph into well-connected clusters~\cite{peng2015partitioning}. 

Currently it is not known whether such strong results also hold in general for higher-order Laplacians, but there are some recent encouraging results in this direction~\cite{parzanchevski2016isoperimetric, horak2013spectra, steenbergen2014cheeger, gundert2014higher}, which have motivated work on spectral clustering for simplicial complexes~\cite{palande2020spectral}. With these results in mind, we suggest that the problem of estimating the number of small eigenvalues of higher-order Laplacians is a well-motivated one. 

\ 

The quantum algorithms that we consider can in some instances be implemented on a non-universal model of quantum computer known as the one clean qubit model, whose associated complexity class is known as $\DQC$ (See Appendix~\ref{app:complexity_classes} for a description of the model and associated complexity class). This essentially rules out the possibility that the problem is $\BQP$-hard, unless $\DQC = \BQP$, but raises the question of whether the problem of estimating Betti numbers is a $\DQC$-hard problem. This makes intuitive sense in light of the fact that the cohomology problem is $\QMA_1$-hard. Indeed, it is often useful to view the class $\DQC$ as an ``averaged'' version of the class $\QMA$: in $\QMA$ we care about the existence of a single input to a quantum circuit that causes it to accept with high probability; in $\DQC$ we care about the average acceptance probability of a quantum circuit over all possible inputs from some Hilbert space. This connection is made explicit by some of the complete problems for each class: for $\QMA$, the canonical complete problem is that of finding the value of the smallest eigenvalue of a Hamiltonian; for $\DQC$, a complete problem is estimating the average of an (exponential) number of low-lying eigenvalues of a Hamiltonian. More generally, many $\QMA$-complete problems are concerned with finding a specific state or energy, whilst many $\DQC$-complete problems focus on finding more ``averaged'' quantities such as (normalized) partition functions, or estimating the Jones polynomial of the trace closure of braids (the exact variant of which is $\sharpP$-hard).

In the next section we prove that the {\sc QBNE} problem is indeed $\DQC$-hard. One might also consider the hardness of the true {\sc BNE} problem. It is in fact possible to show that it is hard for a ``perfect completeness'' version of $\DQC$; however, such a class is not a particularly natural one, and a hardness result for this class may not be of interest (indeed, it is not immediately clear that this class would not be classically simulable) and so we do not focus on that here. 

\

\subsection{{\sc quasi-Betti number estimation} is $\DQC$-hard}
Recall the definition of cochain complexes given in Section~\ref{sec:The k-local cohomology problem}, and their input as a tuple $(V[m],d[k])$. As discussed there, for the Betti number estimation problem to be well defined one needs to be able to be given an efficient description of the input vector space, $V[m] = V = \bigoplus_{p=0}^m V^p$.  This efficient description automatically allows one to efficiently \emph{check} whether a given (basis) vector lies in the space $V^p$ or not, since it involves checking only a polynomial (in $m$) number of constraints, each of which can be checked in polynomial time. For the purpose of describing quantum algorithms for Betti number estimation, however, we will also need to be able to efficiently \emph{sample} elements from a space $V^p$. This won't in general be possible, and so this imposes an additional restriction on the instances of the problem that we consider. We define what we mean by an \emph{efficiently sample-able vector space} below.

\begin{definition}[Efficiently sample-able space]\label{def:efficient_sample}
Let $V = \bigoplus_{p=0}^m V^p \subseteq \mathbb{C}^{2^m}$ be a vector space described by a polynomially long (in $m$) list of constraints. We say that $V$ is \emph{efficiently sample-able} if there exists a polynomial-sized quantum circuit $\mathcal{C}$ that can produce an $m$-qubit state 
\equ{
	\rho_{V^p} := \frac{1}{\dim(V^p)} \sum_{v^p \in V^p} \proj{v^p},
}
starting from the all-zeros state on $O(\poly(m))$ qubits, and where we take the $v^p$ to be some orthonormal basis over $V^p$, labelled by $m$-bit strings. Additionally, we require that $\mathcal{C}$ itself must be constructed in polynomial time by a classical Turing Machine given only the list of constraints defining $V$. For instance, if $V$ is a subspace of the fermionic Fock space of $m$ fermions, then $\rho_{V^p}$ can be taken to be the maximally mixed state over all fermionic states from $V$ with fermion number $p$, encoded in qubits (e.g. via some fermion-to-qubit encoding).
\end{definition}
The requirement that the quantum circuit be constructed by a polynomial-time classical Turing machine ensures that no hardness is hidden in the sampling of the vector space.\footnote{E.g. if the list of constraints was given as some arbitrary {\sc 3SAT} formula, then sampling from the space would be equivalent to sampling from its set of solutions -- something we don't expect to be able to do efficiently.} Note that it might also be sufficient to \emph{approximate} $\rho_{V^p}$ by, say, sampling only approximately uniformly from a basis for $V^p$. For some spaces this might be computationally easier than sampling exactly uniformly, however we won't consider this further. 

\ 

We show that {\sc QBNE} is $\DQC$-hard for $\log$-local\footnote{A $\log$-local operator $B$ acting on $n$ fermionic modes / qubits is one of the form $B = \sum_{j=1}^m B_j$, where each $B_j$ acts on at most $\log(n)$ fermionic modes / qubits, and $m = O(\poly(n))$.} coboundary operators acting on efficiently sample-able spaces $V[m]$ and $l=\Theta(m)$, by constructing a particular instance of a complex for which this problem is hard. Note that the problem can only get harder if we drop the restriction that $V[m]$ is efficiently sample-able, and so that harder version of the problem remains $\DQC$-hard. The constructions presented in Section~\ref{sec:The k-local cohomology problem}  make the proof of this fact quite straightforward.

\begin{thm}\label{theo:dqc1_hardness}
	{\sc QBNE} is $\DQC$-hard for cochain complexes $(V[m],d[\log m])$, $l = \Theta(m)$, and remains hard when $V[m]$ is efficiently sample-able.
\end{thm}

\begin{proof}
	{\sc LLSD} was already shown to be $\DQC$-hard for arbitrary $\log$-local, positive semi-definite Hermitian operators in~\cite{gyurik}.\footnote{Note that there is a claimed proof of this result in~\cite{brandao_thesis}. Unfortunately, the proof contains a (major) mistake. See~\cite{gyurik} for further details on this point.} Our approach will be to define an instance of the {\sc QBNE} problem, such that solving that instance allows one to solve a $\DQC$-hard instance of the {\sc LLSD} problem, which then implies the theorem. To this end, let $A_{\DQC}$ be any ($n$-qubit, $\log$-local, positive semi-definite) Hermitian operator for which the {\sc LLSD} problem is $\DQC$-hard with threshold $b$, precision parameters $\delta$ and $\epsilon$, and success probability $\mu$. 
	
	We use the same construction as in the proof of Theorem~\ref{thmCohomology}, using $A_{\DQC}$ as our starting point. Recall that there the complex $(V[m], d[\log n])$ is of the form
	\[
	    C:\qquad 0\xrightarrow{d} V^{0}\xrightarrow{d} \cdots \xrightarrow{d} V^{n+2}\xrightarrow{d} 0\,,
	\]
	where $V = \bigoplus_{p=0}^{m} V^p$ is a subspace of fermionic Fock space on $m=2(n+1)$ fermionic modes. Our coboundary operator $d : V^p \rightarrow V^{p+1}$ will be a $\log$-local fermionic operator of the form $d = \frac{1}{\sqrt{2}}(a_0^\dagger + b_0^\dagger) B$, where $B$ is the fermionic operator obtained by applying the map \eqref{mapsigmaab} and \eqref{basisfN} to $A_{\DQC}$. Recall that there the Laplacian takes the form
	\equ{
	    \Delta = dd^\dagger + d^\dagger d = B^2.
	}
	As noted in our proof of Theorem~\ref{thmCohomology}, the only states supported in $V^{n+2}$ are states of the form $\ket{1,1}\otimes\ket{\psi_n}$, for $\ket{\psi_n}$ an $n$-fermion state obtained via \eqref{basisfN}. Moreover, by \eqref{eq:same_op}, the operator $B$ restricted to the space $\tilde{V}^n$ spanned by the states $\ket{\psi_n}$ coincides precisely with $A_{\DQC}$. 
	Since $B$ has fermion number $0$ by construction, it necessarily preserves the subspace $\tilde{V}^n$ and so, combined with \eqref{eq:same_op}, we have $\braket{\psi_{n+2} | \Delta | \psi_{n+2}} = \braket{\psi | A_{\DQC}^2 | \psi}$ for all $\ket{\psi} \in V^{n}$, where $\ket{\psi_{n+2}} = \ket{1,1}\otimes \ket{\psi_n}$. Putting everything together, this implies that the eigenvalues of $\Delta|_{V^{n+2}}$ are the squares of the eigenvalues of $A_{\DQC}$.
	
	It follows that the spectral density of $\Delta|_{V^{n+2}}$ below a threshold $t$ is the same as the spectral density of $A_{\DQC}$ below threshold $\sqrt{t}$. Hence, to solve the {\sc LLSD} problem on $A_{\DQC}$ with parameters $b,\delta,\epsilon,\mu$, it suffices to solve the {\sc QBNE} problem for this complex $(V[m],d[\log])$ and $l = n+2 = m/2+1$, with parameters $b',\delta',\epsilon,\mu$ for $b' = b^2$ and $\delta' = 2b\delta + \delta^2$. Since the {\sc LLSD} problem with input $A_{\DQC}$ is $\DQC$-hard, this construction gives a $\DQC$-hard instance of the {\sc QBNE} problem. 
	
	\
    
    The only thing that remains is to verify that the space $V$ is efficiently sample-able as per Definition~\ref{def:efficient_sample}. This is straightforward to see: the space $V$ is a subspace of fermionic Fock space on $m = 2(n+1)$ modes subject to the constraints \eqref{constVhard}. To sample from the space $V^p$, we can classically sample from the set of Hamming weight $p$, $m$-bit strings of the form $n_{a_{0}},n_{b_{0}}\,,\cdots, n_{a_{n}},n_{b_{n}}$ that satisfy those same constraints. This set is easy to sample from uniformly: choose $p$ indices from $0,\dots,n$ uniformly at random (without replacement). If we chose index $0$ then we set either $n_{a_0} = 1$, $n_{b_0} = 1$, or both, each with equal probability. If we set both occupancies to 1, we remove an index from our choice at random (to ensure we create a state with fermion number $p$). From each index $i$ of our remaining indices, we set either $n_{a_i} = 1$ or $n_{b_i} = 1$ with equal probability. 
    
    Each $m$-bit string will correspond to a unique basis state from $V^p$, and so by sampling these uniformly at random and creating the corresponding state $\ket{n_{a_{0}},n_{b_{0}}\,,\cdots, n_{a_{n}},n_{b_{n}}}$ (and ``forgetting'' which state we produced), the quantum state will be precisely 
    \equ{
        \rho_{V^p} = \frac{1}{\dim{V^p}} \sum_{v^p \in V^p} \proj{v^p},
    }
    which is the state we need to produce to efficiently sample from $V^p$.

\end{proof}

\subsection{{\sc quasi-Betti number estimation} is in $\BQP$}\label{sec:quantum_algorithms}
In this section, we present a general algorithm for estimating Betti numbers of complexes defined by a coboundary operator $d$ and a vector space $V$. To keep things simple, we assume that the vector spaces $V$ are subspaces of an $n$-qubit Hilbert space, and that we are given access to boundary operators that are defined in qubit space. This is contrary to the previous sections, however if we are instead given access to boundary operators in \emph{fermionic} space, we note that there exist fermion-to-qubit transformations that can be computed efficiently and give at most a logarithmic overhead in terms of locality of the underlying operators~\cite{bravyi2002fermionic, seeley2012bravyi}, and so this assumption is without loss of generality.

Our algorithm is essentially a simplification of the Lloyd, Garnerone, and Zanardi (LGZ) algorithm of~\cite{LloydetalTDA} that applies to $\log$-local boundary operators, but generalised to apply to a wider class of complexes (i.e. not just the clique complex). If instead the boundary operators are sparse (but not local), then one can solve the {\sc QBNE} problem only when either the Laplacian is guaranteed to have an inverse-polynomial spectral gap above zero, or if it is guaranteed to be sparse itself (and we can efficiently construct sparse access to it). We comment more on this at the end of the section.

\begin{thm}\label{theo:inBQP}
	{\sc QBNE} is in $\BQP$ for cochain complexes $(V[n],d[O(\log n)])$, any $l$, and when $V[n]$ is efficiently sample-able. 
\end{thm}
\begin{proof}

By its definition, solving the {\sc QBNE} problem is the same as solving the {\sc LLSD} problem for the operator $\Delta|_{V^l}$ for $\Delta = (d+d^\dagger)^2$ the Laplacian of the cochain complex. If $d$ is $\log$-local, then clearly $\Delta$ is also $\log$-local, and hence sparse. Quantum algorithms for solving {\sc LLSD} on sparse Hermitian matrices were presented in~\cite{gyurik}; these algorithms involve sampling uniformly random eigenvalues of the target matrix by running phase estimation with an input of the maximally mixed state over the full $n$-qubit Hilbert space. Here, we are only interested in the eigenvalues of the matrix $\Delta|_{V^l}$, and so we can run the same algorithm but instead use the maximally mixed, $n$-qubit state
\equ{
	\rho_{V^l} = \frac{1}{\dim(V^l)} \sum_{v^l \in V^l} |v^l\rangle\langle v^l|
}
over the space $V^l$ as input to the phase estimation routine. By definition, all eigenvectors of $\Delta|_{V^l}$ are fully supported in this space, and so the algorithm run in this way will sample uniformly random eigenvalues of $\Delta|_{V^l}$. Finally, we remark that this state can be constructed in polynomial time by a quantum circuit if the space $V$ is efficiently sample-able in the sense of Definition~\ref{def:efficient_sample}. 

For completeness (and because we will refer to parts of this algorithm later on), we briefly describe the steps of the quantum algorithm for {\sc LLSD} of the matrix $\Delta|_{V^l}$. Informally, the steps are:
\begin{enumerate}
	\item Use Hamiltonian simulation to produce a unitary $U \approx e^{i\Delta}$
	\item Run phase estimation on $U$ with input $\rho_{V^l}$. This produces random samples of (approximations of) the eigenvalues of $\Delta|_{V^l}$, since the maximally mixed state $\rho_{V^l}$ can always be expressed as a uniform mixture over eigenvectors of $\Delta|_{V^l}$. 
	\item Repeat a polynomial number of times, keeping a count of how many eigenvalues are output that fall below the input threshold $b$. 
\end{enumerate}  

From this we can produce an approximation satisfying Equation~\ref{eq:quasi_betti_number_estimation} in polynomial time via a Hoeffding bound. Of course, one must check that the errors arising from the Hamiltonian simulation and phase estimation routines can be bounded appropriately. For a complete description of the algorithm, as well as a full error analysis, we defer to~\cite{gyurik}.
\end{proof}

Combined with Theorem~\ref{theo:dqc1_hardness}, it follows that the complexity of estimating Betti numbers of the family of complexes considered in this paper lies somewhere between $\DQC$ and $\BQP$. Finally, we make the observation that for some instances of the problem, the complexity actually lies inside $\DQC$, which can be seen by making use of a $\DQC$ version of the quantum algorithm for {\sc LLSD} from~\cite{gyurik} as a subroutine. We prove the following in Appendix~\ref{app:proof_of_dqc1_algo}.

\begin{corollary}\label{cor:dqc1_betti_algo}
	{\sc QBNE} is in $\DQC$ for cochain complexes $(V[n],d[O(\log n)])$, any $l$, and when the vector space $V = \bigoplus_{p=0}^n V^p$ satisfies $\frac{\dim(V^l)}{2^n} = \Omega\left(\frac{1}{\poly(n)}\right)$ \emph{and} membership in $V^l$ can be tested using a workspace of at most $O(\log n)$ bits.
\end{corollary}
The requirement that membership in the input Hilbert space can be tested using only a logarithmically many workspace bits/qubits is somewhat unsatisfactory, but seems to be necessary to fit the rejection sampling step of the algorithm into $\DQC$ (see the proof in Appendix~\ref{app:proof_of_dqc1_algo}). In~\cite{aaronson2017computational} Aaronson et al. define the class $\mathsf{TQC}$ (for ``trace computing quantum polynomial time'') and its sampling variant, $\mathsf{sampTQC}$. They find that $\mathsf{TQC} = \DQC$, but that $\mathsf{sampTQC}$  lies somewhere between $\DQC$ and $\BQP$. Informally, the class is defined as any problem that can be solved using a one-clean-qubit computer, with the additional power that one is allowed to learn, after the one-clean-qubit computation, what state from the maximally mixed part of the input was actually given to the circuit. Provided with this information, one can carry out some extra post-processing on the result of the one-clean-qubit computation. With this model, we could perform the rejection-sampling step inside the post-processing part of a $\mathsf{sampTQC}$ computation, meaning that the requirement that membership in the Hilbert space can be checked using only a small workspace can be dropped. This implies the following corollary.
\begin{corollary}
	{\sc QBNE} is in $\mathsf{sampTQC}$ for cochain complexes $(V[n],d[O(\log n)])$, any $l$, and when the vector space $V = \bigoplus_{p=0}^n$ satisfies $\frac{\dim(V^p)}{2^n} = \Omega\left(\frac{1}{\poly(n)}\right)$.
\end{corollary}
This raises the intriguing question of whether Betti number estimation is also \emph{complete} for the class $\mathsf{sampTQC}$. 

\ 

\noindent Finally, we comment on our earlier remark that one can only solve the {\sc QBNE} problem for sparse (but not local) boundary operators when either the Laplacian $\Delta$ itself is guaranteed to be sparse, or is guaranteed to have an inverse-polynomial spectral gap above zero (in which case {\sc QBNE} is equivalent to {\sc BNE}). 

If the Laplacian is itself sparse, then we can efficiently use Hamiltonian simulation to obtain an approximation of $e^{i\Delta}$, and the algorithm can proceed as usual.\footnote{For the clique/independence complex, one can show that the Laplacian actually \emph{is} sparse~\cite{gyurik}.} If it is not sparse, then one has to take a different approach. In particular, it is possible to instead simulate the operator $B = (d + d^\dagger)$, which will be sparse as long as $d$ is. Since the eigenvalues of $\Delta$ are the squares of the eigenvalues of $B$, one would hope that by inputting the maximally mixed state over $V^p$ to phase estimation on the unitary $e^{iB}$, and estimating the fraction of eigenvalues falling below the threshold $b$, it would be possible to obtain an estimate of the low-lying spectral density of $\Delta|_{V^p}$, and hence solve the {\sc QBNE} problem. 

However this is not necessarily the case: in {\sc QBNE} we are interested in the number of (orthogonal) eigenvectors $\ket{\psi^p} \in V^p$ of $\Delta$ such that $\braket{\psi^p | \Delta|_{V^p} | \psi^p} \leq b$ for some threshold $b$. However, any such state $\ket{\psi^p}$ cannot be an eigenvector of $B$ unless $\braket{\psi^p | \Delta|_{V^p} | \psi^p} = 0$ : acting on $\ket{\psi^p}$ with $B = (d + d^\dagger)$ will yield a state with support in both $V^{p+1}$ and $V^{p-1}$, \emph{unless} $d\ket{\psi^p} = 0$ and $d^\dagger\ket{\psi^p} = 0$, in which case $\ket{\psi^p}$ must be a zero-eigenvector of $\Delta|_{V^p}$. Hence, estimating the number of small non-zero eigenvalues of $B$ within the space $V^p$ tells us nothing about the number of small non-zero eigenvalues of $\Delta$ within that space, and hence this approach will not, in general, allow us to solve the {\sc QBNE} problem. 

In some cases however, it \emph{will} work: when $\Delta|_{V^p}$ has a spectral gap of at least $b$ above zero, then any eigenvalues below $b$ must be zero-eigenvalues, and in this case the low-lying spectral density of $\Delta|_{V^p}$ is equivalent to the low-lying spectral density of $B|_{V^p}$, and this quantity just counts the fraction of zero-eigenvalues of $\Delta$, i.e. solves the true {\sc BNE} problem. However, we emphasize that without this promise the method of inputting the maximally mixed state over $V^p$ to phase estimation on the unitary $e^{iB}$ will not allow one to estimate the low-lying spectral density of $\Delta|_{V^p}$, because the non-zero eigenvectors of $B$ and $\Delta$ do not coincide.

\section{Betti Numbers in Topological Data Analysis}
\label{sec:Betti numbers in Topological Data Analysis}

As discussed in the Introduction, the cohomology problem encompasses a large number of computational problems, including ones of practical interest. In this section, we focus on applications to topological data analysis (TDA), specifically to the interesting problem of extracting topological features of large data sets. As we discuss, this amounts to determining the ground states of the fermion hard-core model reviewed in the Introduction. 

Understanding the ground state properties of physical systems (very much including condensed matter systems) is conjectured to be one of, if not the, most promising use-cases of quantum computers~\cite{cao2019quantum, mcclean2021foundations}. Various tools for doing so have been developed over the past few years and we propose that these could be re-purposed for TDA (more generally, for many (co)homology problems). As we argue, it is entirely possible that current state-of-the-art computational techniques for studying condensed matter systems can be applied to this problem.  We first provide an overview of the problem and describe the precise Hamiltonian to be considered, comment on the conditions to obtain a local Hamiltonian, and then discuss implications for the LGZ algorithm in Section~\ref{sec:Existing quantum algorithms for TDA}. Finally, in Section~\ref{sec:variational_TDA} we propose a new variational approach to TDA.

\begin{figure}[t]
\begin{center}
\includegraphics[scale=0.7]{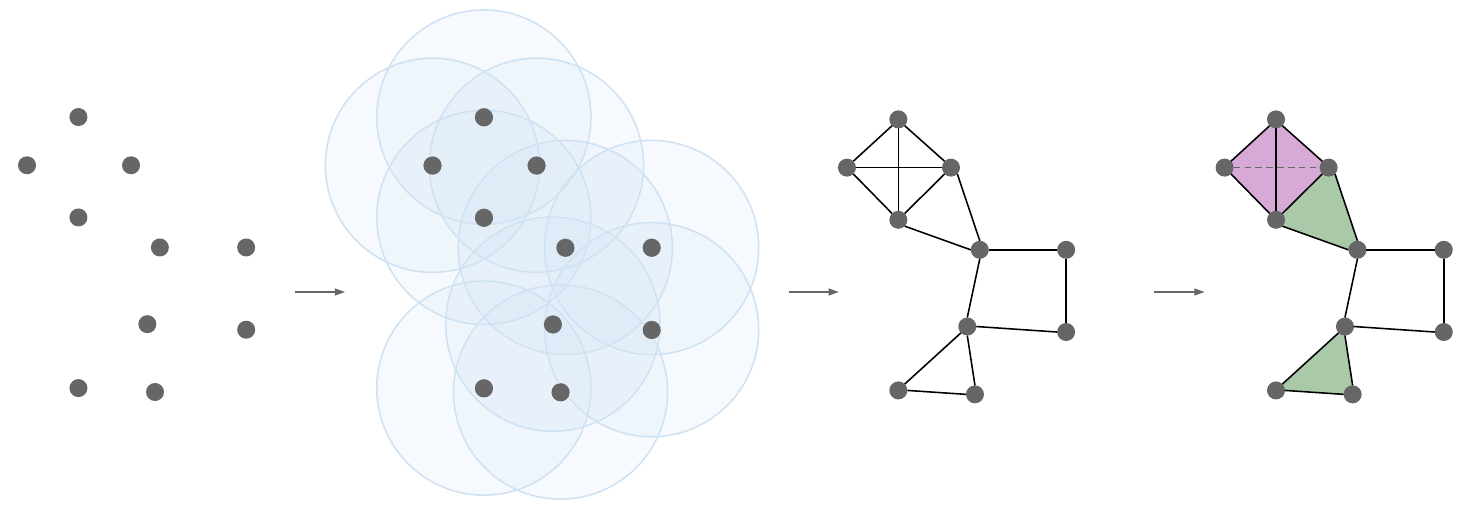}
\end{center}
\caption{An example of the TDA pipeline: one starts with a point cloud (left), adds edges between vertices within a certain distance of each other, resulting in a graph $G$ (middle right), from which the clique  complex is constructed (right). In this example the complex is connected and has a single one-dimensional hole, thus $\beta_0=\beta_1=1$ and all other Betti numbers vanish. }
\label{fig:simplex}
\end{figure}

\subsection{TDA and the fermionic hard-core model}\label{sec:TDA_and_fermion_hardcore}
Topological data analysis is an approach for analyzing data using tools from topology, whose initial motivation was to study the ``shape of data,'' for some mathematically rigorous notion of shape. In this setting one is given a large data set (collection of points in space) and the task is to extract topological features such as the number of $l$-dimensional holes (corresponding to Betti numbers $\beta_l$) that can then be used to classify and analyse it.\footnote{See e.g.~\cite{ghrist2008barcodes} for an introductory overview.} The main tool is so-called ``persistent homology'' -- studying the homology of simplicial complexes derived from point-cloud data in order to look for \emph{persistent} holes in the data, or rather in the topological space that the data is sampled from. A first step towards persistent homology is simply to compute homology, i.e., to compute Betti numbers of a particular simplicial complex, which in this context is usually the \emph{Vietoris-Rips complex}~\cite{ghrist2008barcodes}.

More formally, consider a collection of points $\{x_i\}_{i=1}^n$ embedded in $\mathbb{R}^m$. From this collection of points, one constructs a graph $G$ by associating data points to vertices, and adding an edge between two vertices $i$ and $j$ if they are close, $\|x_i - x_j\|_2 \leq \epsilon$, where $\epsilon$ is termed the \emph{grouping scale}. One then constructs the clique complex of $G$, $C(G)$, a simplicial complex whose simplices are the cliques of the graph $G$ -- see Figure~\ref{fig:simplex}. A complex built in this way is called the Vietoris-Rips complex. Let $C_k(G) \subset \{0,1\}^n$ denote the set of $k$-cliques in $G$, where an element $\ket{s_k} \in C_k(G)$ is represented by a Hamming weight $k$, $n$-bit string whose non-zero entries specify the vertices that constitute the clique. As with any simplicial complex, the boundary operator $\partial$ acts  as
\equ{\label{cobdrysimplex}
    \partial \ket{s_k}= \sum_{l=1}^k (-1)^l \ket{\widehat{s_{k-1}}(l)},
}
where $\ket{\widehat{s_{k-1}}}\in C_{k-1}(G)$ is  obtained by setting the $l$th 1 of $\ket{s_k}$ to 0. That is, the boundary operator maps $k$-cliques to an alternating sum of the $(k-1)$-cliques that it contains (the signs in the sum are determined by choosing a fixed ordering over the simplices of the complex). The homology groups $H_l(\partial)$ are then spanned by $l$-dimensional holes in the complex at a given grouping scale $\epsilon$. Holes that persist over various grouping scales are deemed to be persistent holes, and form the basis for persistent homology. Hence, in order to perform topological data analysis in this way, one needs to determine the homology groups $H_l(\partial)$ or Betti numbers $\beta_l=\text{rank}\, H_l(\partial)$. 

As mentioned, there is a direct connection between TDA in the above sense and the properties of a particular (supersymmetric) many-body  system. Indeed, as discussed in the Introduction, the supersymmetric ground states of the fermionic hard-core model \eqref{HHC} for a graph $G=(V,E)$ precisely captures the cohomology groups of the independence complex $I(G)$. Since we are interested in the clique complex, which is the dual of the independence complex, we consider the model defined on the complement graph $\bar G$ and define the coboundary operator 
\equ{\label{eq:coboundary_clique}
    d = \sum_{i \in \bar{V}} a_i^\dagger P_i, \qquad P_i = \prod_{(i,j)\in \bar E} (1 - \hat{n}_j)\,,
}
where $\bar V$ and $\bar E$ denote the vertex and edge set of $\bar G$. We now note that the operator $d^\dagger$ acts on cliques $\ket{s_k}$ in the same way as the coboundary operator $\partial$ does in \eqref{cobdrysimplex}. We leave as an exercise to verify that the fermionic anti-commutation relations automatically impose the correct signs in the sum of~\eqref{eq:coboundary_clique}, something that results from fixing an ordering over fermionic modes, which in turn defines an orientation of the simplices. On the other hand,  in the absence of torsion coefficients (as is the case here) one has $H_l(d^\dagger)\cong H^l(d)$ and thus 
\equ{H_{l}(\partial)\cong H^{l}(d)\,.
} 
Elements of cohomology are then zero-eigenstates of the Laplacian of the (co)chain complex,  given by
\equ{\label{eq:fhm_laplacian}
    \Delta = \sum_{(i,j) \in \bar{E}} P_i a_i^\dagger a_j P_j + \sum_{i \in \bar{V}} P_i\,,
}
which is simply the supersymmetric Hamiltonian for the fermion hard-core model on the graph $\bar G$ (c.f., Eq. \eqref{HHC}). Thus, the problem of finding $l$-dimensional holes in the clique complex amounts to counting supersymmetric ground states of this system with fermion number $l$. 

In the case of highly structured graphs, the supersymmetric ground states of the fermion hard-core model can be studied analytically or numerically.
For instance, on a square lattice with periodic boundary conditions, the supersymmetric (zero-energy) ground states of the system are found at at fillings in the range $1/5 \leq f \leq 1/4$~\cite{huijse2009supersymmetry} and are in one-to-one correspondence with configurations that arise from tiling the plane using four rhombuses. The ground states are also known for the so-called martini and octagon-square lattices~\cite{fendley2005exact}. For more general graphs, however, the cohomology of the independence complex is much harder to compute. For instance, for the triangular and hexagonal lattices, the ground state structure of the fermion hard-core model is not known, although it is understood that the ground states occur in a finite window of filling~\cite{jonsson2010certain}, and that there is extensive ground state entropy~\cite{van2005extensive,Huijse:2011aa}. For a general graph $G$, however, there is no systematic approach to determining supersymmetric ground states and indeed we expect the problem to be intractable.\footnote{Some partial information about supersymmetric ground states such as the Witten index could be extracted more easily. The Witten index for these systems is $\#\P$-hard, but  finding an additive approximation to the normalized Witten index is in $\BPP$ \cite{Crichigno:2020vue}. }

\subsection{Existing quantum algorithms for TDA}
\label{sec:Existing quantum algorithms for TDA}
A quantum algorithm for TDA that estimates normalized Betti numbers of the clique complex of a graph (in fact it solves the QBNE problem from Definition~\ref{def:QBNE} for clique complexes)  was described by Lloyd, Garnerone, and Zanardi (LGZ)~\cite{LloydetalTDA}. The algorithm is essentially the one outlined in Section~\ref{sec:quantum_algorithms}, in which a maximally mixed state over cliques is used as input to phase estimation on a unitary obtained via Hamiltonian simulation of either the ``Dirac operator'' or Laplacian of the complex. Recent work by Gyurik et al.~\cite{gyurik} considered the possibility of implementing the LGZ algorithm on near-term quantum devices, and suggested that only a modest number of error-corrected qubits ($\approx 80$) might be required to see quantum advantage (i.e. to run the algorithm on input sizes beyond the reach of the best known classical methods) if various optimistic conditions are met.

The computational bottlenecks of the algorithm are the complexities of the subroutines used to implement Hamiltonian simulation of either the combinatorial Laplacian or the Dirac operator $B = (\partial + \partial^\dagger)$ constructed from the boundary operator $\partial$, as well as the cost of generating the maximally mixed state over all $k$-simplices of the complex. The latter is dealt with only by assuming that the graph is sufficiently ``clique-dense'' (which will likely mean that the graph itself is very dense), allowing either rejection sampling or Grover search to be used to create the maximally mixed state.
The Hamiltonian simulation subroutines can be expensive because although the boundary operator and the Laplacian are known to be sparse, they do not have an obvious decomposition into local terms, meaning that one must make use of the more resource-heavy algorithms for sparse Hamiltonian simulation instead of ``lightweight'' techniques for simulating local Hamiltonians. For a detailed discussion of this point, see~\cite{gyurik}.

Let us consider the Hamiltonian \eqref{eq:fhm_laplacian}.  Note that the projectors $P_i$  are nontrivial only when there is no edge between vertex $i$ and vertex $j$ in the input graph $G$, and hence the Laplacian contains terms which act on up to a total of $(n-\delta)$ fermionic modes, with $\delta$ the {\it minimum} degree of the original graph $G$. The Laplacian is thus $k$-local, provided
\equ{\label{deltank}
\delta= n-k\,,
}
with $k$ a constant, i.e., for very edge-dense graphs $G$. Thus, by switching our attention to the coboundary operator of the independence complex, the Laplacian of interest will have a decomposition into local terms provided \eqref{deltank} holds. Moreover, this local decomposition is naturally expressed in terms of fermionic operators, suggesting that it could be possible to make use of modern techniques from the quantum chemistry literature for simulating fermionic Hamiltonians to obtain hardware efficient implementations for the Hamiltonian simulation part of the algorithm.

Graphs that are very edge-dense are likely to be natural candidates for obtaining large quantum speedups using the LGZ algorithm. Since the algorithm is efficient only when the graph $G$ is sufficiently $l$-clique dense, it is likely that the graph is also edge dense -- i.e., all vertices have very high degree. Indeed, a sufficient condition for $l$-clique denseness is when all vertices have degree at least $\gamma n$, for $\gamma > \frac{l-2}{2(l-1)}$~\cite{reiher2016clique}. Hence, a class of graphs for which the LGZ algorithm could exhibit quantum advantage will be when the degree of every vertex is at least $n-k$ for constant $k$, in which case the complement of the graph will have maximum degree $k$, and so the Laplacian \eqref{eq:fhm_laplacian} will be $k$-local. Hamiltonian simulation of local Fermionic operators is likely to be a less resource intensive task on near-to-medium-term devices than simulation of general sparse operators. For example, one can take an efficient Fermion-to-qubit mapping (e.g. the Bravyi-Kitaev transform~\cite{bravyi2002fermionic}) to obtain a $\log$-local qubit operator $\hat{d}$ corresponding to the coboundary operator $d$. Hamiltonian simulation of the $\log$-local qubit operator $\hat{B} = (\hat{d} + \hat{d}^\dagger)$, or of the Laplacian itself, can then be implemented via Trotterization using no extra ancilla qubits. If we have information about the underlying structure of the graph $G$, it is possible that this simulation could be made even more efficient~\cite{childs2019theory}. Moreover, since we are interested only in the low-energy subspace of these operators, it is possible that some recent results (e.g.,~\cite{sahinogluhamiltonian}) regarding Hamiltonian simulation within low-energy subspaces could be exploited to make the simulation even more efficient. We leave these questions open for future work.

\subsection{Alternative quantum and quantum-inspired approaches to TDA }\label{sec:variational_TDA}
The quantum algorithm described in the previous section suffers from a couple of drawbacks. First, it is only efficient when the graph is clique dense, which will likely restrict the family of graphs to which it can be applied. Second, it computes only \emph{normalized} Betti numbers, and it is not entirely clear that these quantities have any immediate practical application. In particular, being able to estimate normalized $l$th Betti numbers up to inverse-polynomial additive error only allows one to distinguish complexes with very many $l$-dimensional holes from those with very few $l$-dimensional holes, which again might limit the family of graphs to which this algorithm can usefully be applied. 

The description of elements of clique/independence homology as ground states of the (supersymmetric) fermion hard-core model reviewed above, a nontrivial quantum mechanical system of interacting fermions, suggests that the problem of clique/independence homology is quantum mechanical in nature. This raises the  prospect of repurposing standard techniques for studying ground states of quantum mechanical systems for TDA, such as VQE or \cite{peruzzo2014variational,mcclean2016theory} or QAOA \cite{farhi2014quantum}. Similarly,  there are are a plethora of \textit{non-quantum} computational tools for approximating ground states of quantum many-body systems, including quantum Monte Carlo and tensor networks~\cite{leblanc2015solutions}, which might yield efficient \emph{classical} algorithms for topological data analysis. One immediate observation however, in both cases, is that often these techniques are designed for systems on very structured graphs (e.g. 2D lattices), and hence they might struggle to perform well on more arbitrary input graphs. Understanding exactly when these classical and quantum techniques perform well is an interesting question, and could shed light on when (if at all) we might expect to obtain quantum advantage for TDA, which we leave for future work.

\section{Discussion and Outlook}
\label{Discussion and Outlook}

We have argued that the intersection of many-body physics, complexity theory, and supersymmetry is a natural arena to study the complexity of problems in homological algebra, in particular the (co)homology problem. Indeed, this perspective reveals that the problems of deciding whether a (co)homology group is nontrivial and of estimating its rank are intrinsically quantum mechanical and serve as natural hard problems for various quantum complexity classes. Precisely, we defined {\sc $k$-local Cohomology} and showed that its complexity lies  between $\QMA_{1}$ and $\QMA$. Similarly, the complexity of an additive approximation to (quasi) Betti numbers lies  between  $\DQC$ and $\BQP$.

These complexity results may come as a surprise at first, as the computational problems we discuss are defined purely in the context of (co)homology and are not manifestly quantum mechanical. As we have observed, however, these {\it can} be seen as quantum mechanical when viewed  as problems in {\it supersymmetric} quantum mechanics, and relate to properties of the ground states of supersymmetric Hamiltonians. Indeed, the complexity of {\sc $k$-local Cohomology} ultimately follows from the fact that the $k$-local Hamiltonian problem for supersymmetric systems remains $\QMA$-complete, as shown in Section~\ref{sec:The SUSY local Hamiltonian problem}.

\

A set of interesting questions arises when considering simpler instances of supersymmetric Hamiltonians. An obvious question, for instance, is whether $\QMA$-hardness of the local  Hamiltonian problem for systems with $\cN=2$ SUSY is maintained for $2\leq k\leq3$. One may also consider systems with higher degrees of supersymmetry, $\cN>2$, and study whether there is a complexity transition at a certain value of $\cN$ for each $k$. It is often the case, at least in the context of supersymmetric QFTs, that beyond a value of $\cN$ the systems become so constrained that the entire class of Hamiltonians is specified by a small set of parameters.\footnote{An example is the case of 4d QFTs with sixteen (Poincar\'e) supercharges, which are specified by a choice of gauge group and one complex coupling constant. These and other maximally supersymmetric systems are also often integrable. Although integrability does not imply tractability in a computational sense (see, e.g., \cite{aaronson2016computational} for a counterexample) it would be interesting to consider the computational hardness of these systems. } 
We note, however, that although systems with extended supersymmetry can be easily constructed in continuum quantum mechanics and quantum field theory,  much less is known about spin systems with extended supersymmetry (see however the $\cN=4$ systems of \cite{2005JPhA...38.5425S}). 
We also stress that although the Hamiltonian for which we have proven hardness is local, in the sense that it is given by a sum of terms acting at most on a constant number of sites, it is in general {\it geometrically} non-local.
It would be interesting to study whether the  SUSY local Hamiltonian problem remains hard for systems constrained by geometric locality or other physical considerations. 
All this applies also to the problem of Betti number estimation. Finally, it would be interesting to study the complexity of these problems in the case of the fermion hard-core model of \cite{Fendley:2002sg}, especially in relation to topological data analysis  as discussed in Section~\ref{sec:Betti numbers in Topological Data Analysis}. In particular, if the local Hamiltonian problem remains $\QMA$-hard for this model, as it is known to be for similar models like Fermi-Hubbard~\cite{o2021electronic}, then this may imply $\QMA_1$-hardness of the cohomology problem for the independence complex, and likely also $\DQC$-hardness for QBNE of clique complexes. The latter would imply that the LGZ algorithm cannot be dequantized unless $\DQC \subseteq \BPP$.

More generally, we note that our hardness constructions give a general approach for mapping hard problems in Hamiltonian complexity into hard problems in algebraic homology. This suggests that the approach could be exploited to elucidate the complexity of problems  beyond those considered in this paper.  In fact, this suggests defining new problems in computational topology altogether, which are natural from a quantum Hamiltonian complexity perspective.  For instance, one may consider the standard $\QCMA$-complete problem of deciding whether a local Hamiltonian has an ``efficiently describable'' ground state, the $\QMA(2)$-complete problem of deciding if there is a ground state with a particular tensor-product structure, or finding the ground energy of stoquastic Hamiltonians. It would be interesting to determine whether these additional constraints correspond to any natural constraints or promises on the corresponding  problems in homological algebra.

On the more practical side, as we discussed in Section~\ref{sec:Betti numbers in Topological Data Analysis}, viewing the problem of topological data analysis as the problem of finding the ground states of a (supersymmetric) Hamiltonian suggests new approaches to this problem.  In particular, this opens the door to using modern techniques from quantum chemistry for topological data analysis and cohomology problems more generally, as long as these can be formulated in terms of a fermionic coboundary operator acting on Fock space. Before the arrival of fault-tolerant quantum computers, it may also be interesting to consider the application of {\it classical} techniques (e.g. traditional tools for studying condensed-matter systems, or even machine learning techniques~\cite{huang2021provably}) for studying ground states of many-body systems to topological data analysis or other (co)homological problems.

\section*{Acknowledgements}
We are grateful to Ronald de Wolf and Kareljan Schoutens for helpful discussions and feedback on an earlier version of this paper, and Vedran Dunjko, Casper Gyurik, Ido Niesen, and Ian Marshall for useful discussions. CC was supported by QuantERA project QuantAlgo 680-91-034. PMC is supported by the European Union’s Horizon 2020 Research Council grant 724659 Massive-Cosmo ERC–2016–COG and the STFC grant ST/T000791/1.

\newpage 

\appendix

\section{Complexity Classes}\label{app:complexity_classes}
In this section, we define the complexity classes $\QMA$, $\QMA_1$, and $\DQC$, which play a central role in our work.

We begin by defining the quantum complexity class  $\QMA$. Loosely speaking, this is the class of problems whose solution can be verified in polynomial time by a quantum computer, a  quantum analog of the class $\NP$. More formally, the quantum complexity class $\QMA[c,s]$ is defined as follows:
\begin{definition}[$\QMA {[c,s]}$]
A promise problem $A = (A_{\text{yes}}, A_{\text{no}})$ is in $\QMA [c,s]$ if there exists a polynomial-time quantum circuit (the ``verifier'') $V_x$ for any $n$-bit input string $x$ such that
\begin{itemize}
	\item If $x \in A_{\text{yes}}$, there exists a $\poly(n)$-qubit witness state $\ket{w}$ such that $\Pr[V_x(\ket{w}) = 1] \geq c$,
	\item If $x \in A_{\text{no}}$, then for any $\poly(n)$-qubit witness state $\ket{w}$, $\Pr[V_x(\ket{w}) = 1] \leq s$.
\end{itemize}
That is,  for a yes-instance, there must exist a quantum witness that causes the verifier to accept with probability $\geq c$. Otherwise, the verifier must reject all witnesses with probability $\geq 1- s$. 
\end{definition}

The parameters $c$ and $s$ are known as the completeness and soundness parameters, respectively.  The class $\QMA[2/3,1/3]$ is simply denoted by $\QMA$ and is viewed as the quantum analogue of the class $\mathsf{NP}$. The choice of $2/3$ and $1/3$ is arbitrary; in fact we have that $\QMA = \QMA[1-1/\exp(n), 1/\exp(n)] = \QMA[c, c-1/\poly(n)]$. A class that will be relevant for us is the ``perfect completeness'' variant of $\QMA$, $\QMA[1,1-1/\poly(n)] = \QMA_1$. Whether $\QMA_1 = \QMA$ is a longstanding open question in quantum complexity theory.\\

\paragraph{A comment on perfect completeness.} The notion of a perfect completeness complexity class introduces some nuances. In particular, the notion of a ``universal gate set'' becomes tricky, and one must fix a gate set in advance and require that any quantum circuits (i.e. the verifier circuit in the case of $\QMA_1$) be constructed from this gate set. It is not known whether the definition of $\QMA_1$ can be made independent of the choice of gate set. To illustrate this point, consider a problem contained in $\QMA_1$ with respect to some choice of universal gate set $\mathcal{G}$, meaning that the verifier circuit, constructed out of gates from $\mathcal{G}$, accepts any valid witness with probability 1. Now consider moving to another universal gate set $\mathcal{G}'$: of course the verifier circuit can still be constructed out of gates from $\mathcal{G}'$, but only \emph{approximately}, even though the approximation can be exponentially close to the target circuit. For perfect completeness, however, this is not enough: we require that the circuit built using $\mathcal{G}'$ also accepts valid witnesses with certainty, and not with ``merely'' a probability exponentially close to 1. Hence, the problem would be considered to be in $\QMA_1$ with respect to $\mathcal{G}$, but \emph{not} in $\QMA_1$ with respect to $\mathcal{G}'$. In~\cite{Gosset_2013}, Gosset and Nagaj use a standard universal gate set $\mathcal{G} = \{\hat{H}, T, CNOT\}$, and for our purposes we choose this gate set also for our definition of $\QMA_1$. For a detailed discussion of this point, we refer the reader to~\cite{Gosset_2013}. 

\ 

Next we define the one-clean-qubit model, and its associated complexity class $\DQC$. Here, we are given access to a \emph{fixed} input 
\begin{equation}
     \rho = \proj{0} \otimes \frac{I_n}{2^n}
\end{equation}
consisting of a single qubit in a pure state $\ket{0}$, and $n$ qubits in a maximally mixed state. We are allowed to apply a polynomial (in $n$) sized circuit, without intermediate measurements, to this state, and then we can measure the clean qubit. We are allowed to perform this procedure a polynomial number of times, and the probability of obtaining a $0$ on the output is considered the output of the model. The associated complexity class is called $\DQC$ and is defined as follows.

\begin{definition}[$\DQC$]\label{def:DQC1}
A language $L = (L_{\text{yes}} L_{\text{no}})$ is in $\DQC$ if, for every $r$-bit input $x \in L$, there exists a quantum circuit $U_x$ acting on $n = \poly(r)$ qubits and consisting of $O(\poly(n))$ quantum gates, such that 
\begin{itemize}
	\item If $x \in L_{\text{yes}}$, \[p_{\text{yes}} = \Tr\left[ \left( U_x \left( \proj{0} \otimes \frac{I}{2^{n-1}} \right) U_x^\dag \right) \proj{1} \otimes I \right] \geq a\]
	\item If $x \in L_{\text{no}}$, \[p_{\text{yes}} = \Tr\left[ \left( U_x \left( \proj{0} \otimes \frac{I}{2^{n-1}} \right) U_x^\dag \right) \proj{1} \otimes I \right] \leq b\]
\end{itemize}
for $a - b \geq 1/O(\poly(n))$.
\end{definition}
The promise gap between $a$ and $b$ means that, to solve any problem in $\DQC$, it is enough to estimate the ``average acceptance probability'' of $U_x$, $p_{\text{yes}}$, up to inverse-polynomial additive error.

\section{$\DQC$ Algorithm for Estimating Quasi Betti Numbers}\label{app:proof_of_dqc1_algo}
Here we provide a proof of Corollary~\ref{cor:dqc1_betti_algo}.

\begin{proof}
	Observe that the condition $\frac{\dim(V^l)}{2^n} = \Omega\left(\frac{1}{\poly(n)}\right)$ immediately implies that $V^l$ is efficiently sample-able, and so the set of instances satisfying this constraint is a subset of those in $\BQP$: if the space makes up a polynomially large fraction of the entire $n$-qubit Hilbert space, one can use rejection sampling to sample uniformly from $V^l$ by repeatedly generating random $n$-bit strings until finding one that is contained in the set $S^l$ of strings that give the computational basis states that span $V^l$. Each string has at least a $1/\poly(n)$ probability of being in $S^l$, and hence this rejection sampling takes polynomial time (and succeeds with probability arbitrarily close to 1). We will show that in the case that membership in $V^l$ can be tested using only logarithmically many ancilla (qu)bits, this rejection sampling step can effectively be performed inside $\DQC$.
	
	First, we note that if the co-boundary operator $d$ is $\log$-local, then the (Hermitian) operator $A = (d + d^\dagger)^2$ that we use for Hamiltonian simulation will also be $\log$-local. Hamiltonian simulation of $\log$-local Hamiltonians can be performed in polynomial time in $\DQC$ using the vanilla simulation technique of Trotterization (\cite{cade2018quantum}, Section~3.4). Hence, as pointed out in~\cite{gyurik}, each step of the algorithm from Theorem~\ref{theo:inBQP} can be performed in $\DQC$, since the Hamiltonian simulation step requires no clean ancilla qubits, and the phase estimation routine requires at most $O(\log n)$ of them. The same error analysis applies in this case as in the proof of Theorem~\ref{theo:inBQP}, and so for the rest of the proof we will assume that each part of the quantum algorithm works perfectly, and instead focus on the problem of doing rejection sampling within the one-clean-qubit model. We will refer to the circuit that performs steps 1 and 2 of the algorithm sketched in the proof of Theorem~\ref{theo:inBQP} (i.e. phase estimation on (an approximation of) $e^{iA}$) as $C$, and assume that it works perfectly -- i.e. for an eigenstate $\ket{\psi}$ of $A$, it outputs the corresponding eigenvalue $\lambda$ perfectly and with certainty.  
	
In $\DQC$, we have at our disposal $O(\log n)$ clean qubits (i.e. in the state $\ket{0}$), and a maximally mixed register of $n$ qubits in the state 
	\[
		\rho = \frac{1}{2^n} \sum_{x \in \{0,1\}^n} \proj{x}.
	\]
	The algorithm of Theorem~\ref{theo:inBQP} requires a maximally mixed state over the space $V^l$, and so we need a way to emulate that input within $\DQC$ given access to only the full maximally mixed state $\rho$. Provided there exists a (classical) circuit to check whether a given bit-string is contained in $S^l$ that uses no more than $O(\log n)$ workspace bits, we can perform this check within the one clean qubit model itself. This allows us to build the rejection sampling step into the $\DQC$ algorithm, in one of two ways.

	If we are allowed to measure 2 qubits at the end of the computation, rather than the usual single qubit, then we can perform the rejection sampling as a post-selection step at the end of the computation. Such a model was considered in, e.g.,~\cite{morimae2014hardness}, although this model does not technically give rise to the same complexity class $\DQC$. Alternatively, we could estimate the quantity $\frac{\dim(V^l)}{2^n}$ classically (i.e. in $\BPP$), and estimate the quantity $\frac{\tilde{\beta_l}}{2^n}$ using a $\DQC$ algorithm, where $\tilde{\beta_l}$ is used to denote the number of eigenvalues of $A_{V^l}$ below the given threshold. By dividing the latter by the former, and ensuring that the estimation accuracies are chosen appropriately, we can obtain an estimate of $\frac{\tilde{\beta_l}}{\dim(V^l)}$, which will give us a solution to the {\sc normalized quasi-Betti number estimation} problem after errors from the circuit are taken into account. The class $\DQC$ is defined\footnote{The accepted definition of $\DQC$ is quite hard to pin down. We take it to be the set of problems solvable in polynomial time given access to a one-clean-qubit quantum computer (which we can run non-adaptively) plus polynomial-time probabilistic classical computation, which feels quite natural -- after all, one must have access to \emph{some} form of classical computation in order to, say, post-process the results of a polynomial number of measurement outcomes in order to output a final answer to the decision problem. This is the view taken also by Brand\~ao in~\cite{brandao_thesis}} to be the set of problems solvable in polynomial time, given access to a one-clean-qubit quantum computer plus polynomial time classical post-processing, and hence this would be a true $\DQC$ algorithm. 

\

We begin by considering the first approach. As described in the proof of Theorem~\ref{theo:dqc1_hardness}, if we were to input the state $\proj{0}_A \otimes \rho_{V^l}$ to $C$, where the first register consists of the $O(\log n)$ ancilla qubits used by the circuit, we would obtain an approximation of a uniformly random eigenvalue of $A_{| V^l}$ in the first register. We can append an extra ancilla qubit and use this as a flag that we set to $\ket{1}$ when the eigenvalue in the first register is $\leq b$, and $\ket{0}$ otherwise. Then the probability of measuring a 1 on this qubit at the end of the computation will be equal to the fraction of eigenvalues of $A_{| V^l}$ that are $\leq b$. However, since $\DQC$ requires us to input the maximally mixed state $\rho$ over the full $n$-qubit Hilbert space we will also obtain estimates of eigenvalues of the other parts of $A$. To remedy this, we can attach another ancilla qubit and use it as a flag that we set to $\ket{1}$ at the beginning of the computation only when the state on the maximally mixed register is a computational basis state from $V^l$ (here is where it is important that this checking only requires $O(\log n)$ many ``work'' qubits, which we must provide as yet more clean ancilla qubits). Finally, if we can measure the two flag qubits at the end of the computation, then by keeping a count of the number of times we see a state from $V^l$, and the number of times we see an eigenvalue $\leq b$, we can estimate the fraction of eigenvalues smaller than $b$ out of all eigenvalues of $A_{| V^l}$ using a Hoeffding bound.  

\

Now we consider the second approach. We proceed in exactly the same way as above, except we add a single extra ancilla qubit in which we store the logical AND of the two flag qubits. We then measure this qubit at the end of the computation as permitted by the canonical definition of $\DQC$. The circuit for this is shown in Figure~\ref{fig:dqc1_circuit}, where we call the unitary for the entire circuit $V$.

\begin{figure}[h]
    \centering
    \includegraphics[scale=0.75]{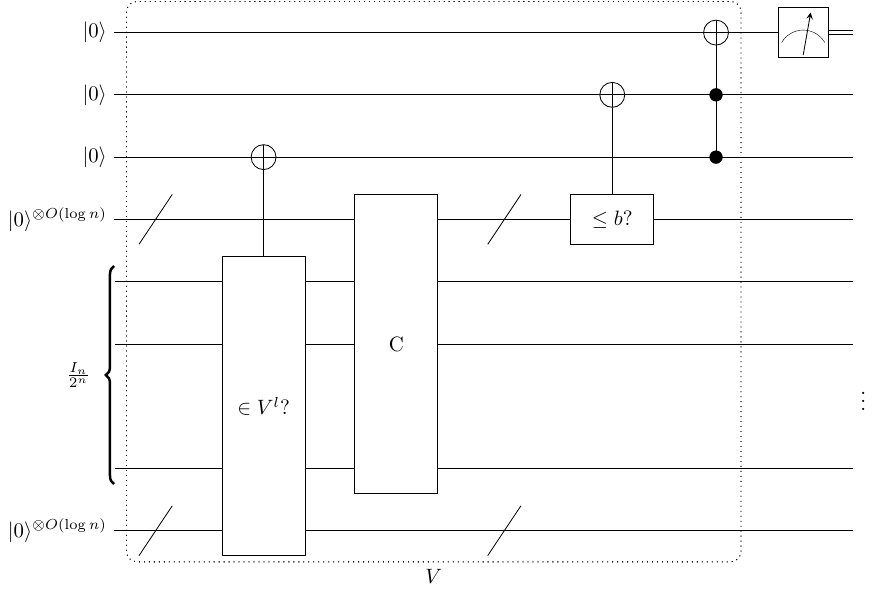}
    \caption{$\DQC$ circuit for quasi Betti number estimation. The gate ``$\in V^l$?'' is a gate that checks (using at most $O(\log n)$ ancillas) whether a (computational basis) state is contained in $V^l$ or not, QPE is quantum phase estimation (for some choice of unitary), and the ``$\leq b$?'' gate checks whether the value stored in the phase register is smaller than $b$ or not.}
    \label{fig:dqc1_circuit}
\end{figure}

Collecting the clean qubits into a single register $A$, the probability of measuring a $1$ on the first qubit is
\begin{equation} 
	\Pr[\text{measure }1] = \Tr\left[ V\left(\proj{0}_A \otimes \frac{I}{2^n}\right)V^\dagger (\proj{1} \otimes I) \right].
\end{equation}
Since $(d + d ^\dagger)^2$ is block-diagonal with each block acting on a separate $V^l$, every one of its eigenvectors, and therefore those of $A$, are either entirely supported within some $V^l$, or else not supported at all. By expanding the equation above and taking the maximally mixed state to be a uniform mixture over eigenvectors of $A$ (plus any extra orthogonal vectors required to complete a basis for $\mathbb{C}^{2^n}$), we can see that the probability of measuring a 1 at the end of the circuit is just the probability that some input state $\ket{\psi}$ chosen from the ensemble is contained in $V^l$ and that this state is an eigenstate of $A$ with eigenvalue $\leq b$. Hence, the probability of measuring 1 at the end of the computation, taken over all such states $\ket{\psi}$, is
\begin{eqnarray}
	\Pr_{\ket{\psi}}[\text{measure } 1] &=& \Pr_{\ket{\psi}}[\ket{\psi} \in V^l] \cdot \Pr_{\ket{\psi}}[\braket{\psi|A|\psi} \leq b | \ket{x} \in V^l] \\
	&=& \frac{\dim(V^l)}{2^n} \cdot \frac{\tilde{\beta_l}}{\dim(V^l)} = \frac{\tilde{\beta_l}}{2^n}.
\end{eqnarray}
By repeating the computation a polynomial number of times, we can obtain an estimate of this probability up to additive accuracy $1/\poly(n)$. Hence, we can obtain a quantity 
\begin{equation}\label{eq:betti_number_estimate}
	\frac{\tilde{\beta_l}}{2^n} \pm \frac{1}{p(n)},
\end{equation}
for $p(n)$ some polynomial in $n$. 

Next, we can classically estimate the fraction of all $n$-qubit (computational basis) states that lie inside $V^l$ -- i.e. estimate the relative dimension of the space $V^l$. To do this we can sample $n$-bit strings uniformly at random and check whether they satisfy the description of $V^l$ provided as input. Once again, in polynomial time and by a Hoeffding bound we can obtain an estimate
\begin{equation}\label{eq:hilbert_space_size}
	\frac{\dim(V^l)}{2^n} \pm \frac{1}{r(n)}
\end{equation}
for some polynomial $r(n)$. Finally, dividing the quantity in~\ref{eq:betti_number_estimate} by the one in~\ref{eq:hilbert_space_size}, we obtain an estimate
\begin{eqnarray}
	\left( \frac{\tilde{\beta_l}}{2^n} \pm \frac{1}{p(n)} \right) \left( \frac{\dim(V^l)}{2^n} \pm \frac{1}{r(n)} \right)^{-1} 
	&=& \left( \frac{\tilde{\beta_l}}{2^n} \pm \frac{1}{p(n)} \right) \frac{2^n}{\dim(V^l)} \left( 1 \pm \frac{2^n}{r(n)\dim(V^l)} \right)^{-1} \\
	&=& \left( \frac{\tilde{\beta_l}}{\dim(V^l)} \pm \frac{2^n}{p(n)\dim(V^l)} \right) \left( 1 \pm \frac{2^n}{r(n)\dim(V^l)} \right)^{-1}. \nonumber
\end{eqnarray}
Using the fact that $\frac{\dim(V^l)}{2^n} \geq 1/q(n)$ for some polynomial $q(n)$, our estimate satisfies
\begin{eqnarray}
	 \left( \frac{\tilde{\beta_l}}{\dim(V^l)} \pm \frac{q(n)}{p(n)} \right) \left( 1 \pm \frac{q(n)}{r(n)} \right)^{-1}.
\end{eqnarray}
Choosing $r(n)$ and $p(n)$ sufficiently large so that $\frac{q(n)}{p(n)} = \epsilon = 1/\poly(n)$ and $\frac{q(n)}{r(n)} = 1/\poly(n) = \epsilon'$, we have an estimate 
\begin{eqnarray}\label{eq:annoying_error}
	 \left( \frac{\tilde{\beta_l}}{\dim(V^l)} \pm \epsilon \right) \left( 1 \pm \epsilon' \right)^{-1}.
\end{eqnarray}
We can then see that
\[
	(1 \pm \epsilon')^{-1} = 1 \mp \frac{\epsilon'}{1 \pm \epsilon'} = 1 \pm \frac{1}{O(\poly(n))},
\]
and hence the estimate in Equation~\ref{eq:annoying_error} satisfies
\begin{eqnarray}
	 \frac{\tilde{\beta_l}}{\dim(V^l)} \pm \frac{1}{O(\poly(n))}
\end{eqnarray}
for some polynomial depending on $\epsilon$ and $\epsilon'$. 

\

In the end we obtain an estimate of $\frac{\tilde{\beta_l}}{\dim(V^l)}$ up to inverse polynomial additive error in polynomial time using a $\DQC$ algorithm, proving the statement.
\end{proof}

\bibliography{references.bib} 
\bibliographystyle{plainnat}

\end{document}